\documentclass[12pt]{iopart}

\usepackage{amsfonts}
\usepackage{graphicx}  
\usepackage{iopams}  
\usepackage{hyperref}

\newtheorem{theorem}{Theorem}[section]
\newtheorem{lemma}[theorem]{Lemma}

\newtheorem{corollary}[theorem]{Corollary}

\newenvironment{proof}[1][Proof]{\begin{trivlist}
\item[\hskip \labelsep {\bfseries #1}]}{\end{trivlist}}

\newcommand{\qed}{\nobreak \ifvmode \relax \else
      \ifdim\lastskip<1.5em \hskip-\lastskip
      \hskip1.5em plus0em minus0.5em \fi \nobreak
      \vrule height0.75em width0.5em depth0.25em\fi}

\begin{document}

\title[Odd Parity Perturbations of the Self-Similar LTB Spacetime]{Odd Parity Perturbations of the Self-Similar LTB Spacetime}

\author{Emily M. Duffy and Brien C. Nolan}

\address{School of Mathematical Sciences, Dublin City University, Glasnevin, Dublin 9, Ireland.}
\eads{\mailto{emilymargaret.duffy27@mail.dcu.ie}, \mailto{brien.nolan@dcu.ie}}
\begin{abstract}
We consider the behaviour of odd-parity perturbations of those self-similar Lema\^{i}tre-Tolman-Bondi spacetimes which admit a naked singularity.  We find that a perturbation which evolves from initially regular data remains finite on the Cauchy horizon. Finiteness is demonstrated by considering the behaviour of suitable energy norms of the perturbation (and pointwise values of these quantities) on natural spacelike hypersurfaces. This result holds for a general choice of initial data and initial data surface. Finally, we examine the perturbed Weyl scalars in order to provide a physical interpretation of our results. Taken on its own, this result does not support cosmic censorship; however a full perturbation of this spacetime would include even parity perturbations, so we cannot conclude that this spacetime is stable to all linear perturbations. 

\end{abstract}

\submitto{\CQG}

\maketitle

\section{Introduction: Cosmic Censorship and Stability}
\label{sec:intro}
Typically, during the course of the gravitational collapse of a massive body, a black hole event horizon forms, after which the imploding body collapses into a singularity which is hidden behind the event horizon. However, in certain models the order of horizon and singularity formation is reversed and the singularity is free to communicate with the external universe. Such naked singularities are generally considered to be an undesirable feature of gravitational collapse, as they effectively destroy the classical predictability of a spacetime. In particular, associated with the naked singularity is a Cauchy horizon which represents a barrier past which the physical evolution of matter is not predictable. On a physical level naked singularities could potentially emit unlimited amounts of matter and energy.

In response to the presence of naked singularities in certain collapse models, Penrose hypothesised that nature should always censor the singularity associated with gravitational collapse. In other words, this cosmic censorship hypothesis (CCH) loosely states that physically realistic gravitational collapse will not result in naked singularity formation \cite{Wald}. Although mathematically rigorous versions of the CCH exist, it has not been proven. It is generally expected that some form of the CCH will hold for physically reasonable spacetimes. Nonetheless, studies of naked singularity spacetimes are useful, as they provide a means of examining the nature of the putative cosmic censor. Furthermore, they have resulted in more precise formulations of the CCH and may provide insight into how it might be proven. 

One manner in which a naked singularity spacetime may not be a serious counterexample to the CCH is if the Cauchy horizon associated with the naked singularity is not stable. In this case, perturbations which begin their evolution as regular functions on some initial surface would diverge on the Cauchy horizon. The naked singularity would then be regarded as a single (non-typical) member of a whole class of spacetimes, in which the Cauchy horizon is replaced with a null singularity. 

Should perturbations of a given naked singularity spacetime behave in a finite manner at the Cauchy horizon, one can in some cases still rule out this spacetime as a serious counter-example to the CCH due to other defects. In particular, certain naked singularity spacetimes arise from unrealistic matter models which can form singularities even in flat spacetime (for example, the Vaidya spacetime displays this property). Other naked singularity spacetimes display sensitivity to the choice of initial conditions, in that if the initial conditions are slightly perturbed, the naked singularity fails to form. We can neglect such spacetimes as serious counter-examples to the CCH. 

We consider here the self-similar Lema\^{i}tre-Tolman-Bondi spacetime. This is a spherically symmetric spacetime in which a pressure-free perfect fluid collapses inhomogeneously into a singularity. A spacetime displays self-similarity if it admits a homothetic Killing vector field, that is, a vector field $\vec{\xi}$ such that 
\begin{equation} \label{selfsim}
\mathcal{L}_{\vec{\xi}} g_{\mu \nu} = 2 g_{\mu \nu}, 
\end{equation}
where $\mathcal{L}$ indicates the Lie derivative. The choice of non-zero constant on the right hand side is arbitrary, and can be fixed by rescaling $\vec{\xi}$. Imposing the existence of a homothetic Killing vector field on a spherically symmetric spacetime results in considerable simplification. The metric of a spherically symmetric spacetime can always be written in the form 

\begin{equation*} 
ds^2=-\rme^{2\Phi}(t,r) dt^2 + \rme^{2 \Psi}(t,r) dr^2 + R^2(t,r) d \Omega^2,
\end{equation*}
where $\Phi(t,r)$ and $\Psi(t,r)$ are arbitrary functions of $t$ and $r$ and $d \Omega^2= d \theta^2 + \sin^2\theta d \phi^2$ is the usual metric on a two-sphere. Imposing condition (\ref{selfsim}) results in the scalings
\begin{eqnarray*} 
\Phi(t,r)=\Phi(z), \qquad \qquad 
\Psi(t,r)=\Psi(z), \qquad \qquad 
R(t,r)=r S(z),
\end{eqnarray*}
where $z=-t/r$. See \cite{CC} for a discussion of the role of self-similarity in general relativity. See \cite{CahillTaub} and \cite{Eardley} for a discussion of spacetimes with homothetic Killing vector fields.  

We note that the self-similar LTB spacetime cannot be taken as a serious counter-example to the CCH, as the first of the abovementioned defects is present in this spacetime. The matter model used is dust, which ignores pressure (and pressure gradients), and therefore cannot be expected to provide a realistic description of gravitational collapse. One could also reasonably expect that this matter model would break down during the collapse to the singularity, as the curvature of the spacetime becomes extreme. Nonetheless, the simplicity of this spacetime makes it a very useful toy model of  gravitational collapse resulting in naked singularity formation. 

We also note that a natural application of this work would be the study of odd parity perturbations of the self-similar perfect fluid spacetime. This spacetime has been considered by \cite{CCGNU}. We also note the work of Harada and Maeda \cite{HaradaMaeda} which discusses the role of the self-similar solution in the general spherically symmetric collapse of a perfect fluid. Indeed this paper in part motivates the current work. Harada and Maeda have found numerical evidence of a stable naked singularity in soft fluid collapse. In this paper, we develop techniques that we hope will allow us to study the linear stability of Harada and Maeda's naked singularity spacetime in a rigorous mathematical fashion.

In this paper, we consider linear odd parity perturbations of the self-similar LTB spacetime. We choose to study the odd parity perturbations first as the even parity perturbations obey a far more complex system of equations which presents extra technical difficulties. We find that the odd parity perturbations remain bounded up to and on the Cauchy horizon with respect to certain energy norms defined on natural spacelike hypersurfaces. One may be tempted to interpret this result as evidence against the cosmic censorship hypothesis; however, a full treatment of the perturbations of this spacetime would have to include the linear even parity perturbations, and indeed, non-linear perturbations as far as possible. We therefore cannot interpret our result concerning the behaviour of odd parity perturbations only as providing evidence against cosmic censorship. 

We are currently considering the behaviour of even parity perturbations of this spacetime using methods broadly similar to those of this paper, and these results will be discussed in a future paper. Even parity perturbations of this spacetime have been studied both numerically \cite{HaradaMaeda} and analytically \cite{LTBeven}. Harada et al found numerical evidence of instability in the $\ell=2$ perturbation. In \cite{LTBeven}, both a harmonic decomposition and a Fourier mode decomposition of the gauge invariant perturbation variables and equations were used, and the individual Fourier modes $X_{\omega,\ell,m}(z)$ were analysed. It was found that modes which are finite on the past null cone remain finite on the Cauchy horizon. However, the question of how to resum these modes on the Cauchy horizon was not fully addressed, and so the problem of getting a complete analytic understanding the even parity linear perturbations remains unresolved.

We note also that should linear perturbations turn out to grow without bound as the Cauchy horizon is approached, the perturbative framework would not be valid and a full non-linear analysis would be required. 

In the next section, we describe the structure of the self-similar LTB spacetime and examine the conditions necessary for this spacetime to contain a naked singularity. In section 3, we discuss a gauge invariant perturbation formalism for spherically symmetric spacetimes due to Gerlach and Sengupta. We specialise this formalism to the odd partity case, and show that the matter perturbation can be entirely determined by a choice of an initial data function. In the odd parity case, the metric perturbation can be described by a single gauge-invariant scalar whose evolution is determined by the linearised Einstein equations. These equations reduce to a single inhomogeneous wave equation in the gauge invariant scalar, sourced by an initial data function. In section 4, we first present an existence and uniqueness result for solutions to this equation before showing that this scalar remains finite on the Cauchy horizon, subject to a specification of initial data. Finiteness is measured in terms of naturally occuring energy norms evaluated on suitable hypersurfaces which are generated by the homothetic Killing vector field $\vec{\xi}$. This result also applies to the first derivatives of the perturbation scalar. This indicates that there is no instability present at the level of the metric or matter perturbation. In section 5 we provide a physical interpretation of our results in terms of the perturbed Weyl scalars. In section 6 we make some concluding remarks. We follow throughout the example of \cite{vaidya} and \cite{scalar} and we use units in which $G=c=1$. 


\section{The Self-Similar LTB Spacetime }
\label{sec:background}
\subsection{The LTB Spacetime }
\label{sec:LTB}
The Lema\^{i}tre-Tolman-Bondi spacetime is a spherically symmetric spacetime containing a pressure-free perfect fluid which undergoes an inhomogeneous collapse into a singularity. Under certain conditions this singularity can be naked. We will initially use comoving coordinates $(t,r, \theta, \phi)$, in which the dust is stationary so that the dust velocity has a time component only. In these coordinates, the radius $r$ labels each successive shell in the collapsing dust. The line element for such a spacetime can be written in comoving coordinates  as 

\begin{equation}
ds^2=-dt^2 + \rme^{\nu}(t,r)dr^2 + R^2(t,r)d\Omega^2,
\label{metric_tr}
\end{equation}
where $d\Omega^2=d\theta^2+\sin^2\theta d\phi^2$ and $R(t,r)$ is the physical radius of the dust. The stress-energy tensor of the dust can be written as 

\begin{equation*}
\bar{T}^{\mu \nu} = \bar{\rho}(t,r) \bar{u}^{\mu} \bar{u}^{\nu},
\end{equation*}
where $\bar{u}^{\mu}$ is the 4-velocity of the dust, that is, a future pointing, timelike unit vector field, which is tangential to the flow lines of the dust and satisfies $\bar{u}_{\mu} \bar{u}^{\mu}=-1$. $\bar{\rho}(t,r)$ is the rest mass density of the dust. In comoving coordinates, $\bar{u}^{\mu} = \delta_{0}^{\mu}$. 

The background Einstein equations for the metric and stress energy in comoving coordinates immediately provide the following results:
\begin{eqnarray}
\fl \rme^{\nu /2}=\frac{R'}{\sqrt{1+f(r)}} , \qquad   \bar{\rho}(t,r)=\frac{m'(r)}{4 \pi R' R^2} ,  \qquad \left(\frac{\partial R}{\partial t} \right)^2-\frac{2 m(r)}{R}=f(r),
\label{bgnd}
\end{eqnarray}  
where $'=\frac{\partial }{\partial r}$. The function $m(r)$ is known as the Misner-Sharp mass and is a suitable mass measure for spherically symmetric spacetimes.  The last equation in (\ref{bgnd}) has the form of a specific energy equation, which indicates that the function $f(r)$ can be interpreted as the total energy per unit mass of the dust. The background dynamics of the dust cloud can be determined  by a choice of $m(r)$ (or a specification of the initial profile of $\rho(t,r)$) and a choice of $f(r)$. 

Recall that a shell focusing singularity is a singularity which occurs when the physical radius $R(t,r)$ of the dust cloud vanishes, so that all the matter shells have been ``focused'' onto a single point. In this spacetime, a shell focusing singularity occurs on a surface of the form $t=t_{sf}(r)$, which includes the scaling origin $(t,r)=(0,0)$. 

In spacetimes consisting of a collapsing cloud of matter, one can also encounter a shell crossing singularity, which occurs when two shells, labelled by particular values of the radius, $r_{1}$ and $r_{2}$, cross each other. More precisely, there are values $r_{1}$, $r_{2}$ and times $t_{A}$, $t_{B}$ for which $R(t_{A}, r_{1}) < R(t_{A}, r_{2})$ but $R(t_{B}, r_{1}) > R(t_{B}, r_{2})$. No such singularity occurs in the spacetime under consideration here \cite{ES}. 

We immediately specialise to the marginally bound case by setting $f(r)=0$. 


\subsection{Self-Similarity}
\label{sec:selfsim}
We follow here the conventions of \cite{CC}. In comoving coordinates, the homothetic Killing vector field is given by $\vec{\xi}=t\frac{\partial}{\partial t}+ r \frac{\partial}{\partial r }$. When self-similarity is imposed on the metric and stress-energy tensor, we find that functions appearing in the metric, the dust density and the Misner-Sharp mass have the following scaling behaviour:
\begin{eqnarray} \label{nuandR}
\nu(t,r)=\nu(z), \qquad \qquad \qquad
R(t,r)=rS(z),
\end{eqnarray}
\begin{eqnarray} \label{rhoandm}
\bar{\rho}(t,r)=\frac{q(z)}{r^2}, \qquad \qquad \qquad
m(r)=\lambda r,
\end{eqnarray}
where $z=-t/r$ is the similarity variable and $\lambda$ is a constant (the case $\lambda=0$ corresponds to flat spacetime). By combining (\ref{bgnd}), (\ref{nuandR}) and (\ref{rhoandm}) we can find an expression for $\dot{R}$,
\begin{equation*}
\frac{\partial R}{\partial t}=-\frac{dS}{dz}=-\sqrt{\frac{2 \lambda }{S}},
\end{equation*}
where we choose the negative sign for the square root, so that we are dealing with a collapse model. 

This can be immediately solved for $S(z)$:
\begin{equation} \label{Sdef}
S(z)= (a z +1)^{2/3}, 
\end{equation}
where $a=3\sqrt{\frac{\lambda}{2}}$ and we use the boundary conditions $R|_{t=0}=r$ and $R'|_{t=0}=1$. With this expression for $S(z)$ we can solve for $R'$ explicitly. In (\ref{bgnd}) we convert $R'$ to a derivative in $(z,r)$ and find that 

\begin{equation} \label{expform}
\rme^{\nu/2}=R'=(\frac{1}{3}az+1)(1+az)^{-1/3}. 
\end{equation}
We state the metric in $(z,r)$ coordinates, for future use: 

\begin{equation} \label{metriczr}
ds^2=-r^2 dz^2+\rme^{\nu}(z)(1-z^2\rme^{-\nu}(z))dr^2-2rz drdz + R^2d\Omega^2.  
\end{equation} 
In section \ref{sec:weyl} we will need the null directions of the self-similar LTB spacetime. In terms of $(z,r)$ coordinates, the retarded null coordinate $u$ and the advanced null coordinate $v$ take the form
\begin{eqnarray} \label{uvcoords}
u=r \, \exp \left( - \int_{z}^{z_{o}} \frac{dz'}{f_{+}(z')} \right), \qquad \quad
v=r \, \exp \left( - \int_{z}^{z_{o}} \frac{dz'}{f_{-}(z')}\right),
\end{eqnarray}
where $f_{\pm}:=\pm \rme^{\nu/2} + z$. In these coordinates, the metric takes the form

\begin{equation*} 
ds^2=- \frac{t^2}{uv} (1-\rme^{\nu}z^{-2}) \, du \, dv + R^2(t,r) d \, \Omega^2. 
\end{equation*}
In  order to calculate the perturbed Weyl scalars, we will need the in- and outgoing null vectors, $l^{\mu}$ and $n^{\mu}$. These vectors obey the normalisation $g_{\mu \nu} l^{\mu} n^{\nu}=-1$. A suitable choice is therefore
\begin{equation}
\label{inoutvectors}
\vec{l}=\frac{1}{B(u,v)} \frac{\partial}{\partial u}, \qquad \qquad \vec{n}=\frac{\partial}{\partial v}, 
\end{equation}
where $B(u,v)=\frac{t^2}{2uv} \left( 1-\frac{\rme^{\nu(z)}}{z^2} \right)$. In what follows, we shall take a dot to indicate differentiation with respect to the similarity variable $z$, $\cdot = \frac{\partial}{\partial z}$. 


\subsection{Nakedness of the Singular Origin}
\label{sec:ns}
We now consider the conditions required for the singularity at the scaling origin $(t,r)=(0,0)$ to be naked. As a necessary and sufficient condition for nakedness, the spacetime must admit causal curves which have their past endpoint on the singularity. It can be shown \cite{NW} that it is actually sufficient to consider only null geodesics with their past endpoints on the singularity, and without loss of generality, we restrict our attention to the case of radial null geodesics (RNGs). The equation which governs RNGs can be read off the metric (\ref{metric_tr}): 

\begin{equation*} 
\frac{dt}{dr}=\pm \rme^{\nu /2}. 
\end{equation*}
Since we wish to consider outgoing RNGS we select the $+$ sign. We can convert the above equation into an ODE in the similarity variable:
\begin{equation}\label{zrrngs}
z+rz'=-\rme^{\nu /2}. 
\end{equation}
We look for constant solutions to this equation, which correspond to null geodesics that originate from the singularity. It can be shown that the existence of constant solutions to (\ref{zrrngs}) is equivalent to the nakedness of the singularity. For constant solutions, we set the derivative of $z$ to zero and combine (\ref{expform}) and (\ref{zrrngs}) to find the following algebraic equation in $z$:
\begin{equation*} 
az^4+\left(1+\frac{a^3}{27}\right)z^3+\left(\frac{a^2}{3}\right)z^2 +az+1=0. 
\end{equation*}
We wish to discover when this equation will have real solutions. This can easily be found using the polynomial discriminant for a quartic equation, which is negative when there are two real roots. In this case we have
\begin{equation*}
D=\frac{1}{27}(-729+2808a^3-4a^6),
\end{equation*}
which is negative in the region $a<a^*$ where $a^*$ is
\begin{equation*}
a^*=\frac{3}{(2(26+15\sqrt{3}))^{1/3}} \approx 0.638...
\end{equation*}
This translates to the bound $\lambda \leq 0.09$. From (\ref{rhoandm}), we can see that this result implies that singularities which are ``not too massive'' can be naked. See figure \ref{Fig1} for a Penrose diagram of this spacetime. 
\\ \\
\textbf{Remark 2.1:} In fact, one can find $D < 0$ in two ranges, namely $a<a^* \approx 0.64$ and $a > a^{**} \approx 8.89$. We reject the latter range as begin unphysical. Consider (\ref{Sdef}), which indicates that the shell-focusing singularity occurs at $z=-1/a$. If we chose the range $a>a^{**}$ we would find that the corresponding outgoing RNG occurs after the shell focusing singularity and so is not part of the spacetime. 
\\ \\
\textbf{Remark 2.2:} We note that this analysis has assumed that the entire spacetime is filled with a dust fluid. A more realistic model would involve introducing a cutoff at some radius  $r=r_{*}$, after which the spacetime would be empty. We would then match the interior matter-filled region to an exterior Schwarszchild spacetime. However, it can be shown that this cutoff spacetime will be globally naked so long as the cutoff radius is chosen to be sufficiently small \cite{Joshi}. We will therefore neglect to introduce such a cutoff. 

\begin{figure} 
\begin{center}
\includegraphics[scale=1]{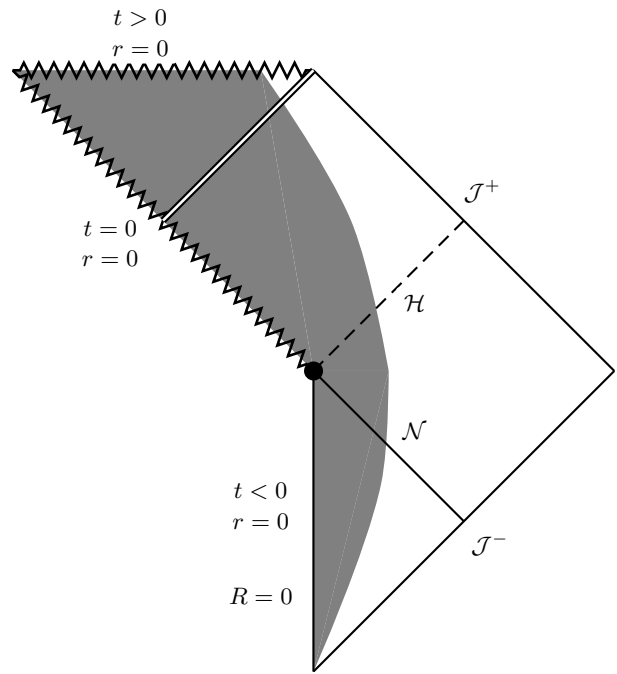}
\end{center}
\caption{Structure of the Self-Similar LTB spacetime. We present here the conformal diagram for the self-similar LTB spacetime. The gray shaded region represents the interior of the collapsing dust cloud. We label the past null cone of the naked singularity by $\mathcal{N}$, future and past null infinity by $\mathcal{J}^{+}$ and $\mathcal{J}^{-}$. 
}
\label{Fig1}
\end{figure}


\section{The Gerlach-Sengupta Formalism} 
\label{sec:GS}
We shall use the Gerlach-Sengupta method \cite{GS} to perturb this spacetime (we follow the presentation of \cite{MGGundlach}). This method exploits the spherical symmetry of the spacetime by performing a decomposition of the spacetime into two submanifolds (with corresponding metrics). Perturbations of the spacetime are then expanded in a multipole decomposition and gauge invariant combinations of the perturbations are constructed. 

We begin by writing the metric of the entire spacetime $(\mathcal{M}^4, g_{\mu\nu})$ as
\begin{equation} \label{gsmetric}
ds^2=g_{AB}(x^{C}) dx^A dx^B + R^2(x^{C}) \gamma_{ab} dx^adx^b, 
\end{equation}
where $g_{AB}$ is a Lorentzian metric on a 2-dimensional manifold $\mathcal{M}^2$ and $\gamma_{ab}$ is the metric on the 2-sphere $\mathcal{S}^2$ 
(and the full manifold is $\mathcal{M}^4=\mathcal{M}^2 \times \mathcal{S}^2$). The indices $A, B, C...$ indicate coordinates on $\mathcal{M}^2$ and take the values $A, B...=0,1$ while the indices $a, b, c...$ indicate coordinates on $\mathcal{S}^2$ and take the values $a,b...=2,3$. The covariant derivatives on $\mathcal{M}^4$, $\mathcal{M}^2$ and $\mathcal{S}^2$ are denoted by a semi-colon, a vertical bar and a colon respectively. The stress-energy can be split in a similar fashion:
\begin{equation*}
t_{\mu \nu}dx^{\mu} dx^{\nu}=t_{AB}dx^Adx^B + Q(x^C)R^2\gamma_{ab}dx^adx^b,
\end{equation*}
where $Q(x^C)=\frac{1}{2}t^a_{\, \, a}$ is the trace across the stress-energy on $\mathcal{S}^2$, which vanishes in the LTB case. Now if we define 

\begin{equation*}
v_{A}=\frac{R_{|A}}{R},
\end{equation*}

\begin{equation*}
V_{0}=-\frac{1}{R^2} + 2v^{A}_{\, \, |A} + 3 v^{A}v_{A},
\end{equation*}
then the Einstein equations for the background metric and stress-energy read 

\begin{equation*}
G_{AB}=-2(v_{A|B}+ v_{A}v_{B})+V_{0}g_{AB}=8 \pi t_{AB},
\end{equation*}

\begin{equation*}
\frac{1}{2}G^{a}_{\, \, a}=-\mathcal{R}+v^{A}v_{A}+v^{A}_{\, \, |A}=8 \pi Q(x^{C} ),
\end{equation*}
where $G^{a}_{\, \, a}=\gamma^{ab} G_{ab}$ and $\mathcal{R}$ is the Gaussian curvature of $\mathcal{M}^2$, $\mathcal{R}=\frac{1}{2} R_{A}^{(2)A}$ where $R^{(2)}$ indicates the Ricci tensor on $\mathcal{M}^2$.  

We now wish to perturb the metric (\ref{gsmetric}), such that $g_{\mu \nu}(x^{\delta}) \rightarrow g_{\mu \nu}(x^{\delta})+ \delta g_{\mu \nu}(x^{\delta})$. To do this, we will use a similar decomposition for $\delta g_{\mu \nu}(x^{\delta})$ and write explicitly the angular dependence using the spherical harmonics. We write the spherical harmonics as $Y^{m}_{l}\equiv Y$. $\{Y\}$  forms a basis for scalar harmonics, while $\{Y_{a}:=Y_{:a}, S_{a}:=\epsilon^{\, \, b}_{a}Y_{b}  \}$ form a basis for vector harmonics. Finally, $\{Y\gamma_{ab}, Z_{ab}:=Y_{a:b}+\frac{l(l+1)}{2}Y \gamma_{ab}, S_{a:b}+S_{b:a} \}$ form a basis for tensor harmonics.

We can classify these harmonics according to their behaviour under spatial inversion $\vec{x} \rightarrow -\vec{x}$. A harmonic with index $l$ is even if it transforms as $(-1)^{l}$ and odd if it transforms as $(-1)^{l+1}$. According to this classification, $Y$, $Y_{a}$ and $Z_{ab}$ are even, while $S_{a}$ and $S_{(a:b)}$ are odd. 

We now expand the metric perturbation in terms of the spherical harmonics. Each perturbation is labelled by $(l,m)$ and the full perturbation is given by a sum over all $l$ and $m$. However, since each individual perturbation decouples in what follows, we can neglect the labels and summation symbols. The metric perturbation is given by

\begin{equation*}
\delta g_{AB}=h_{AB} Y,
\end{equation*}
\begin{equation*}
\delta g_{Ab}=h^{\scriptsize{\textbf{E}}}_{A}Y_{:b}+h^{\scriptsize{\textbf{O}}}_{A}S_{b},
\end{equation*}
\begin{equation*}
\delta g_{ab}=R^2KY \gamma_{ab}+R^2GZ_{ab}+h(S_{a:b}+S_{b:a}),
\end{equation*}
where $h_{AB}$ is a symmetric rank 2 tensor,  $h^{\scriptsize{\textbf{E}}}_{A}$ and $h^{\scriptsize{\textbf{O}}}_{A}$ are vectors and $K$, $G$ and $h$ are scalars, all on $\mathcal{M}^2$. We similarly perturb the stress-energy $t_{\mu \nu} \rightarrow t_{\mu \nu}+ \delta t_{\mu \nu}$ and expand the perturbation in terms of the spherical harmonics: 

\begin{equation} \label{m2se}
\delta t_{AB}=\Delta t_{AB} Y,
\end{equation}
\begin{equation} \label{m2s2se}
\delta t_{Ab}=\Delta t^{\scriptsize{\textbf{E}}}_{A} Y_{:b} + \Delta t^{\scriptsize{\textbf{O}}}_{A} S_{b},
\end{equation}
\begin{equation} \label{s2se}
\delta t_{ab}=r^2 \Delta t^3 \gamma_{ab} Y + r^2 \Delta t^2 Z_{ab} + 2\Delta t S_{(a:b)},
\end{equation}
where $\Delta t_{AB}$ is a symmetric rank 2 tensor,  $\Delta t^{\scriptsize{\textbf{E}}}_{A}$ and $\Delta t^{\scriptsize{\textbf{O}}}_{A}$ are vectors and $\Delta t^3$, $\Delta t^2$ and $\Delta t$ are scalars, all on $\mathcal{M}^2$. 

We wish to work with gauge invariant variables, which can be constructed as follows. Suppose the vector field $\vec{\xi}$ generates an infinitesimal coordinate transformation $\vec{x} \rightarrow \vec{x'}=\vec{x}+ \vec{\xi}$. We wish our variables to be invariant under such a transformation. We can decompose $\vec{\xi}$ into even and odd harmonics and write the one-form fields 

\begin{equation*}
\underline{\xi}^{\scriptsize{\textbf{E}}}=\xi_{A}(x^C)Ydx^A + \xi^{\scriptsize{\textbf{E}}}(x^C)Y_{:a} dx^a,
\end{equation*}
\begin{equation*}
\underline{\xi}^{\scriptsize{\textbf{O}}}=\xi^{\scriptsize{\textbf{O}}} S_{a} dx ^a. 
\end{equation*}
We then construct the transformed perturbations after this coordinate transformation and look for combinations of perturbations which are independent of $\vec{\xi}$, and therefore gauge invariant. We will list here only the odd parity gauge independent perturbations. The odd parity metric perturbation can be written as a gauge invariant vector field:  
\begin{equation} 
\label{ka}
k_{A}=h^{\scriptsize{\textbf{O}}}_{A}-h_{|A}+2hv_{A},
\end{equation}
and the gauge invariant matter perturbation is given by a 2-vector and a scalar: 
\begin{equation} \label{la}
L_{A}=\Delta t^{\scriptsize{\textbf{O}}}_{A}-Qh^{\scriptsize{\textbf{O}}}_{A},
\end{equation}
\begin{equation} \label{lscal}
L=\Delta t - Qh. 
\end{equation}
The linearised Einstein equations which govern the evolution of these perturbations are 
\begin{equation} \label{kaein}
k^A_{\, \, |A}=16 \pi \mathcal{L}, \qquad \qquad l\geq 2,
\end{equation}

\begin{equation} \label{dabein}
(R^4 D^{AB})_{|B}+\mathcal{L}k^A=16 \pi R^2 L^A, \qquad \qquad l\geq 1,
\end{equation}
where $\mathcal{L}=(l-1)(l+2)$ and  $D_{AB}$ is

\begin{equation*} 
D_{AB}=\left(\frac{k_B}{R^2}\right)_{|A}-\left(\frac{k_A}{R^2}\right)_{|B}. 
\end{equation*}
By taking a derivative of (\ref{dabein}), using the fact that $D_{AB}$ is antisymmetric and combining the result with (\ref{kaein}), one can derive the stress-energy conservation equation:

\begin{equation} \label{secons}
(R^2L^A)_{|A}=\mathcal{L} L. 
\end{equation}
One can show that (\ref{dabein}) is equivalent (for $l \geq 2$) to a single scalar equation
\begin{equation} \label{mastergi}
\left( \frac{1}{R^2} (R^4 \Psi)^{|A}  \right)_{|A} - \mathcal{L} \Psi = -16 \pi \epsilon^{AB} L_{A|B},
\end{equation}
where $\epsilon_{AB}$ is the Levi-Civita tensor on the two dimensional manifold, and the scalar $\Psi$ is defined, for $l \geq 2$, by
\begin{equation*} 
\Psi=\epsilon^{AB} (R^{-2}k_{A})_{|B}. 
\end{equation*}
The gauge invariant metric perturbation $k_{A}$ can be recovered from 
\begin{equation} \label{psistoka}
\mathcal{L}k_{A} = 16 \pi R^2 L_{A} - \epsilon_{AB} (R^4 \Psi)^{|B}. 
\end{equation}
In what follows, we will use the Regge-Wheeler gauge, in which $G=h=0$, which implies that that the bare perturbations coincide with the gauge invariant terms. We will henceforth assume that $l \geq 2$. 


\subsection{The Matter Perturbation}
\label{sec:matter}
To proceed further, we must find a relation between the gauge invariant matter perturbation $L^A$ and the dust density and velocity discussed in section \ref{sec:LTB}. To do this, we write the stress-energy of the full spacetime as a sum of the background stress-energy and the perturbation stress-energy (where a bar indicates a background quantity):
\begin{equation*}
T_{\mu \nu}=\overline{T}_{\mu \nu}+\delta T_{\mu \nu}. 
\end{equation*}
We will assume that the full stress-energy of the perturbed spacetime also represents dust. We can write the density as $\rho=\overline{\rho}+\delta \rho$ and the fluid velocity as $u_{\mu}=\overline{u}_{\mu}+ \delta u_{\mu}$. We can therefore find an expression for the perturbed stress-energy (keeping only first order terms):
\begin{equation} \label{dustsepert}
\delta T_{\mu \nu}=\overline{\rho}(\overline{u}_{\mu} \delta u_{\nu}+ \overline{u}_{\nu} \delta u_{\mu})+ \delta \rho \overline{u}_{\mu} \overline{u}_{\nu}. 
\end{equation}
The perturbation of the dust velocity can now be expanded in terms of the spherical harmonics as $\delta u_{\mu}=( \delta u_A Y,  \delta u_{o} S_a )=(0, 0, U(t,r)S_{a})$. If we set all even perturbations in (\ref{m2se} - \ref{s2se}) to zero, then comparison of (\ref{m2se}) and (\ref{s2se}) to (\ref{dustsepert}) produces the results:
\begin{equation*} 
\delta \rho =0,  \qquad \qquad \qquad \Delta t=0.
\end{equation*}
Then comparing (\ref{m2s2se}) to (\ref{dustsepert}) and using (\ref{la}) and (\ref{lscal}) (remembering that $Q=0$ in this spacetime) produces
\begin{equation*} 
L_A=\Delta t^O_A=(\overline{\rho} U, 0), \qquad \qquad \qquad L=0.
\end{equation*} 
If we use these results in conjunction with (\ref{secons}) (noting that the relevant perturbation Christoffel symbols all vanish for the case of odd perturbations in the Regge-Wheeler gauge), we find that (\ref{secons}) becomes:
\begin{equation}  \label{udotfull}
U,_{t}+U \left( \frac{2 R,_{t}}{R}+ \frac{\bar{\rho},_{t}}{\bar{\rho}} + \frac{\nu,_{t}}{2} \right)=0. 
\end{equation}
Conservation of stress-energy on the background spacetime results in 
\begin{equation}  \label{zeroorder}
\bar{\rho},_{t}+\bar{\rho} \left(\frac{2 R,_{t}}{R}+\frac{\nu,_{t}}{2} \right)=0. 
\end{equation}
Combining (\ref{udotfull}) and (\ref{zeroorder}) produces
\begin{equation} \label{udot}
\frac{\partial U}{\partial t}=0. 
\end{equation}
Given this result, the matter perturbation can be completely determined by a choice of initial profile $U(z=z_{i}, r)=y(r)$ on some suitable initial data surface $z_{i} \in (z_{c}, z_{p}]$, where $z_{p}$ indicates the past null cone of the scaling origin. We now exploit these results to find a useful form for (\ref{dabein}).  


\subsection{The Master Equation}
\label{sec:master}
Having specified the matter perturbation in terms of an initial data function, we now consider the remaining odd parity terms. We use the coordinates $(z,p)$ where $z=-t/r$ is the similarity variable introduced above and $p=\ln r$ is a useful scaling of the radial coordinate. In terms of these coordinates, (\ref{mastergi}) can be written as 

\begin{equation} \label{master}
\fl \beta(z) \frac{\partial^2 A}{\partial z^2}+ \gamma(z) \frac{\partial^2 A}{\partial p^2} + \xi(z)\frac{\partial^2 A}{\partial z \partial p} + a(z) \frac{\partial A}{\partial z}+b(z) \frac{\partial A}{\partial p}+ c(z) A = \rme^{\kappa p}\Sigma(z,p), 
\end{equation}
where the function $A(z,p)$ is related to the master function by 
\begin{equation*}
A(z,p)=\rme^{\kappa p} S^{4}(z) \Psi(z,p). 
\end{equation*}
We introduce a factor $\rme^{\kappa p} = r^{\kappa}$, for $\kappa \geq 0$, for reasons which will be explained later. This means that $\Psi$ can be non-zero at the singularity. In what follows, we will find a positive value $\kappa^{*}$ such that $\kappa \in [0, \kappa^*]$.  
The coefficients in (\ref{master}) depend only on $z$ and are given in appendix A. We note that the three leading coefficients are all metric functions, see (\ref{metriczr}). The source term $\Sigma(z,p)$ is 
\begin{equation} \label{Sigmadef}
\Sigma(z,p)=-16 \pi \rme^{-\nu/2}S^{2} \partial_{r}(\overline{\rho}U). 
\end{equation}


\begin{figure} 
\begin{center}
\includegraphics[scale=1.3]{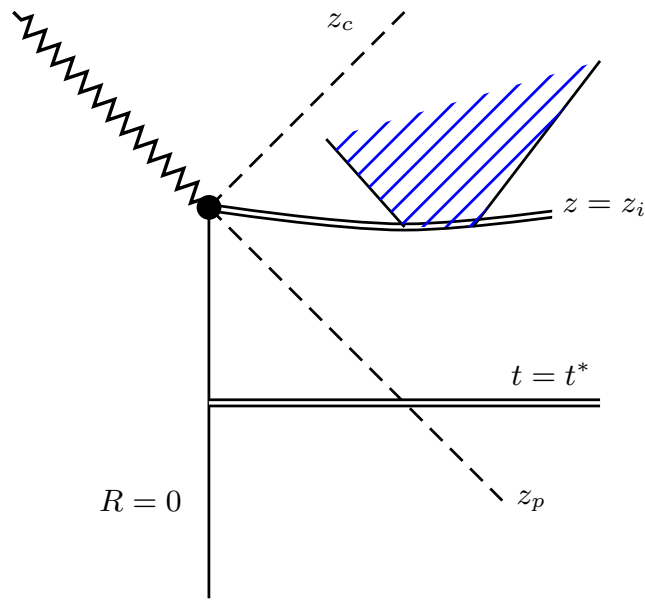}
\end{center}
\caption{The Cauchy Problem. We illustrate here the Cauchy problem associated with the evolution of the perturbation $A(z,p)$ from the initial surface. The support of the initially smooth perturbation (indicated by stripes) spreads causally from the initial surface $z=z_{i}$ up to the Cauchy horizon. 
}
\label{Fig2}
\end{figure}


\subsection{Existence and Uniqueness of Solutions}
\label{sec:existence}

We briefly note that the choice of coordinate $z=-t/r$ means that in the range $z_{c}<z \leq z_{p}$, where $z_{c}$, $z_{p}$ are the Cauchy horizon and past null cone respectively, $z$ is a time coordinate. Since the Cauchy horizon $z=z_{c}$ actually occurs at some negative $z$-value, we should always integrate from $z=z_{c}$ up to $z$. Notice that (\ref{metriczr}) indicates that the Cauchy horizon occurs at $z_{c}=- \rme^{\nu(z_{c})/2}$, while the past null cone of the scaling origin occurs at $z_{p}=+\rme^{\nu(z_{p})/2}$. 

We wish to prove that there exist unique solutions to the initial value problem comprised of (\ref{master}) with suitable initial conditions. We begin by showing that (\ref{master}) may be written as a first order symmetric hyperbolic system. We define a useful coordinate transformation: 
\begin{equation*} 
\overline{z}:=\int_{z}^{z_{i}} \frac{ds}{\beta(s)}, 
\end{equation*}
where $z_{i}$ labels the initial data surface. By inspection, we can see that $\bar{z}(z_{i})=0$. Also, we can see that $\bar{z}(z_c)=\infty$ if we note that we can write $\beta(z) =  z_{c}^{-2}(z_{c}+z)(z_{c}-z) $, so that $z_{c}$ is a simple root of $\beta(z)$. We now define the vector $\vec{\Phi}$:
\begin{center} 
\[ \vec{\Phi}=\left( \begin{array}{c}
A \\
A,_{\bar{z}}+\xi(z) A,_{p} \\
A,_{p} \end{array} \right). \] 
\end{center}
Then (\ref{master}) takes the form
\begin{equation} \label{firstorderred}
\vec{\Phi},_{\bar{z}}=X \vec{\Phi},_{p}+W \vec{\Phi} + \vec{j}. 
\end{equation}
where the matrices $X$ and $W$, and the vector $\vec{j}$ are given in appendix B. In this appendix, we use standard hyperbolic PDE theory to put the system (\ref{firstorderred}) in the form required for the theorem \ref{Thm1}. We identify the surface $S_{i}=\{(z_{i}, p) | \  z_{i}=0, p \in \mathbb{R}\}$ as our initial data surface. Since we have writen (\ref{master}) using self-similar coordinates for the region between $(z_{p}, z_{c})$, this is a suitable choice. $C_{0}^{\infty}(\mathbb{R}, \mathbb{R})$ is the space of smooth functions with compact support. 


\begin{theorem}\label{Thm1}
Let $\vec{f}$ and $\vec{j} \in C_{0}^{\infty}(\mathbb{R}, \mathbb{R}^3)$. Then there exists a unique solution $\vec{\Psi}(z, p)$, $\vec{\Psi} \in C^{\infty}(\mathbb{R}\times(z_{c}, z_{i}], \mathbb{R}^3)$, to the initial value problem consisting of (\ref{firstorderred}) with the initial condition $\vec{\Psi}|_{z_{i}}=\vec{f}$. For all $z \in (z_{c}, z_{i}] $ the vector function $\vec{\Psi}(z, \cdot): \mathbb{R} \rightarrow \mathbb{R}^3$ has compact support. 
\end{theorem}
\begin{proof}
See \cite{McOwen} for a standard proof of this theorem. 
\hfill $\square$
\end{proof}
As a corollary to this theorem, the second order master equation, (\ref{master}), inherits existence and uniqueness. We note here that $z$ is decreasing from $z_{i}$. 


\begin{corollary}
\label{Cor1}
Let $f$, $g$, $\Sigma \in C_{0}^{\infty}(\mathbb{R}, \mathbb{R})$. Then there exists a unique solution $A \in C^{\infty}(\mathbb{R}\times(z_{c}, z_{i}], \mathbb{R})$, to the initial value problem consisting of (\ref{master}) with the initial conditions
\begin{eqnarray*}
A|_{z_{i}}=f \qquad \qquad \qquad
A,_{z}|_{z_{i}}=g
\end{eqnarray*}
For all $z \in (z_{c}, z_{i}] $ the function $A(z, \cdot): \mathbb{R} \rightarrow \mathbb{R}$ has compact support. 
\end{corollary}
This corollary ensures existence and uniqueness for solutions to (\ref{master}) in the region between the initial data surface and the Cauchy horizon. In other words, when $A(z,p)$ has regular initial data, the evolution of $A(z,p)$ remains smooth from the initial data surface up to the Cauchy horizon. However, this does not imply smooth behaviour of the perturbation on the Cauchy horizon. 


\section{Behaviour of Perturbation on Cauchy Horizon}
\label{sec:chbehaviour}
Having given a theorem which guarantees the existence of solutions to (\ref{master}), we now outline the problem under consideration. We insert an initial perturbation from the set of initial data $C_{0}^{\infty}(\mathbb{R}, \mathbb{R})$ on the surface $z=z_{i}$. We then evolve this perturbation up to the Cauchy horizon. We aim to determine whether or not the perturbation remains finite as it impinges on the Cauchy horizon. See figure \ref{Fig2} for an illustration of this. 

We begin by noting that the abovementioned choice of initial data is not ideal. Our choice of initial data surface is dictated by the self-similar nature of the background spacetime, and thus, is a natural choice to make. However, this surface intersects the singular scaling origin $(t=0, r=0)$ of the spacetime. We are therefore forced to consider initial data which is compact supported away from the naked singularity. However, by establishing certain bounds on the behaviour of solutions to (\ref{master}) with this initial data choice, we can then exploit the nature of the space $C_{0}^{\infty}(\mathbb{R}, \mathbb{R})$ to extend these bounds to a more satisfactory choice of initial data which can be non-zero at the scaling origin. 

Finally, we note that since the leading coefficient in (\ref{master}), $\beta(z)$, vanishes on the Cauchy horizon, the Cauchy horizon is a singular hypersurface for this equation. This means that the question of the behaviour of $A(z,p)$ and its derivatives as we approach the Cauchy horizon is nontrivial. To examine this behaviour, we use energy methods for hyperbolic systems. 
 
 
\subsection{First Energy Norm}
\label{sec:e1}
We begin our analysis of the Cauchy horizon behaviour of the perturbation by introducing the energy integral
\begin{equation} \label{defE1}
E_{1}(\bar{z})=E_{1}[A](\bar{z})=\int_{\mathbb{R}} \Vert \vec{\Phi} \Vert^2 dp,
\end{equation} 
where $\Vert \cdot \Vert$ indicates the Euclidean norm. The notation
\begin{equation*}
\Vert \vec{f} \Vert^2_2=\int_{\mathbb{R}} \Vert \vec{f} \Vert^2 dp
\end{equation*}
indicates the $L^2$-norm (squared) of the vector function $\vec{f}(z, p)$. We can immediately state a bound on this energy integral, which is a standard result for equations of the form of (\ref{master}).
 

\begin{corollary}
\label{Cor2}
$E_{1}[A](\bar{z}$) is differentiable on $[0, \infty)$ and satisfies the bound
\begin{equation*} 
E_{1}[A](\bar{z}) \leq \rme^{B_{0}\bar{z}} \left( E_{1}[A](0)+\int_{0}^{\infty} \Vert \vec{j} \Vert^2_2 dp \right), 
\end{equation*}
where $B_{0}=\sup_{\bar{z} > 0} \vert I-2W \vert < \infty$, where $W$ is the matrix appearing in (\ref{firstorderred}). As a consequence, the following results also hold: 
\begin{equation} \label{e1abound}
\int_{\mathbb{R}} \vert A(z, p) \vert^2 dp \leq \rme^{B_{0}\bar{z}} \left( E_{1}[A](0)+\int_{0}^{\infty} \Vert \vec{j} \Vert^2_2 dp \right),
\end{equation}

\begin{equation} \label{e1apbound}
\int_{\mathbb{R}} \vert A,_{p}(z, p) \vert^2 dp \leq \rme^{B_{0}\bar{z}} \left( E_{1}[A](0)+\int_{0}^{\infty} \Vert \vec{j} \Vert^2_2 dp \right),
\end{equation}

\begin{equation} \label{e1azbound}
\int_{\mathbb{R}} \vert A,_{z}(z, p) \vert^2 dp \leq C_{1}\rme^{C_{0}\bar{z}} \left( E_{1}[A](0)+\int_{0}^{\infty} \Vert \vec{j} \Vert^2_2 dp \right),
\end{equation}
where $C_{0}$ and $C_{1}$ are constants, not necessarily equal, which depend only on the angular number $l$ and the metric functions $\nu(t,r)$ and $R(t,r)$. 
\end{corollary}
\begin{proof} This is a standard result which follows from the definition of $E_{1}(\bar{z})$, see \cite{McOwen}. 
\hfill $\square$
\end{proof}
These results indicate that this energy norm is bounded by a divergent term, since $\rme^{B_{0} \bar{z}}$ diverges on the Cauchy horizon. In other words, this theorem only ensures that the growth of the energy norm as we approach the Cauchy horizon is subexponential. In order to proceed, we define a second energy integral, whose behaviour near the Cauchy horizon can be more strongly controlled. 


\subsection{Second Energy Integral}
\label{sec:e2}
Define
\begin{equation} \label{defe2}
E_{2}[A](z):=\int_{\mathbb{R}} \beta(z)A,_{z}^2 - \gamma(z) A,_{p}^2 + H(z) A^2 + K(z)\rme^{2 \kappa p}\Sigma^2(z,p) \, dp, 
\end{equation}
where $H(z)=c(z)$ and $K(z)$ is an arbitrary, non-negative smooth function defined on $(z_{c}, z_{i}]$ which will be fixed later. In Corollary \ref{Cor3}, we will find a useful range for $\kappa$ and in Lemma \ref{Lem1}, we establish that with $\kappa$ in this range, $H(z) \geq 0$, and thus, $E_{2}[A](z) \geq 0$. Using these results, we can control the behaviour of $dE_{2}/dz$ and with this in place, we can finally bound $E_{2}(z)$. 


\begin{lemma}
\label{Lem1}
In the region $\kappa \in [0, \kappa^*]$, where  $\kappa^*:=\frac{9}{4}$, $H(z) \geq 0$ and $\dot{H}(z) \leq 0$, for all $z \in (z_{c}, z_{i}]$. 
\end{lemma}
\begin{proof} We first note that the range $\kappa \in [0, \kappa^*]$ arises from a bound in the next corollary, Corollary \ref{Cor3}. Recall
\begin{eqnarray}
\label{Hdef}
\fl H(z):=c(z)= -\rme^{-\nu}(\kappa^2-5\kappa+4)+z\rme^{-\nu}\left( \frac{\dot{\nu}}{2} + \frac{2 \dot{S}}{S} \right)(\kappa -4) + \mathcal{L}S^{-2}. 
\end{eqnarray}
From (\ref{expform}) and (\ref{Sdef}), we can find explicit forms for each function involved in this definition. We note that $\mathcal{L}$ enters with a coefficient $S^{-2}(z)$, which is always positive. We can therefore safely set $\mathcal{L}=4$, since if $H(z)$ is positive for $\mathcal{L}=4$, it will become larger, and therefore more positive, for larger values of $\mathcal{L}$. So, with $\mathcal{L}=4$, we find that: 
\begin{eqnarray*}
H(z)=\frac{-m(z)-n(z)+p(z)}{(1+az)^{4/3}(3+az)^{3}},
\end{eqnarray*}
where $m(z)=9\kappa^2(3+az)(1+az)^2$, $n(z)=8az(36+45az+13a^2z^2)$ and $p(z)=9\kappa(15+39az+31a^2z^2+7a^3z^3)$. The denominator is clearly positive, as can be verified by explicitly checking the allowed ranges of $a \in (0, a^*)$, $z \in (z_{p}, z_{c})$ and $\kappa \in [0, \kappa^*]$. We next consider the numerator. It can easily be confirmed that for $a$, $\kappa$ and $z$  in their respective ranges, the numerator above is also positive. 

If we consider (\ref{Hdef}), and take a derviative with respect to $z$, we note that the term containing $\mathcal{L}$ will be $-\frac{4}{3} S^{-3/2}(z)$, where we used the fact that $\dot{S}(z)=\frac{2}{3} a S^{-1/2}(z)$. In other words, $\mathcal{L}$ enters with a coefficient which is always negative. So, if we set $\mathcal{L}=4$, and can show that in this case, $\dot{H}(z) \leq 0$, then increasing $\mathcal{L}$ will result in $\dot{H}(z)$ becoming more negative. So, calculating the derviative of $H(z)$ and setting $\mathcal{L}=4$, we find 
\begin{eqnarray*}
\dot{H}(z)=\frac{4a(o(z)-t(z)+u(z))}{3(1+az)^{7/3}(3+az)^4}, 
\end{eqnarray*}
where $o(z)=azm(z)$, $t(z)=9k(-9-3az+27a^2z^2+28a^3z^3+7a^4z^4)$ and $u(z)=4(-162 -243az-63a^2z^2+60a63z^3+26a^4z^4)$. Again, the denominator is clearly positive, and using the same ranges for $a$, $\kappa$, $z$ and $\mathcal{L}$ we can verify that the numerator is negative. So overall, for $a \in (0, a^*)$, $\kappa \in [0, \kappa^*]$, $z \in (z_{p}, z_{c})$ and $\mathcal{L} \geq 4$ (which corresponds to $l \geq 2$), we have $H(z) \geq 0$ and $\dot{H}(z) \leq 0$. 
\hfill$\square$
\end{proof}
We can now move on to examine the behaviour of the derivative of $E_{2}(z)$. 


\begin{corollary}
\label{Cor3}
Let $\kappa \in [0, \kappa^*]$, where $\kappa^*=\frac{9}{4}$. Then there exists some $z^*$ with $z_{c} < z^* \leq z_{i}$, a positive constant $\mu$ and a choice of function $K(z)$ such that $E_{2}(z) \geq 0$ and the derivative of the second energy integral obeys the bound
\begin{equation*} 
\frac{dE_{2}}{dz} \geq - \mu E_{2}(z)
\end{equation*}
in the range $z \in (z_{c}, z^*]$. 
\end{corollary}
\begin{proof} From the definition of $E_{2}(z)$, (\ref{defe2}),
\begin{eqnarray*}
\fl \frac{dE_{2}}{dz}=\int_{\mathbb{R}} ( \beta,_{z} A,_{z}^2 &+ 2\beta A,_{z}A,_{zz} - \gamma,_{z}A,_{p}^2  \\ \nonumber 
&  - 2 \gamma A,_{p}A,_{pz} + H,_{z}A^2+2HA,_{z}A +  K,_{z}\rme^{2 \kappa p}\Sigma^2 + 2 K \rme^{2 \kappa p} \Sigma \Sigma,_{z} ) dp. 
\end{eqnarray*}
We now take the following steps. We remove the term containing $A,_{p}A,_{pz}$ by integrating by parts; the resulting surface term will vanish due to the compact support of $A$. We then replace the term containing $A,_{zz}$ using (\ref{master}). Finally, we remove the term containing $A,_{z}A,_{zp}$ as it is a total derivative. Having followed these steps, we are left with

\begin{eqnarray*}
\fl \frac{dE_{2}}{dz}=\int_{\mathbb{R}} (\beta,_{z}-2a(z))A,_{z}^2-\gamma,_{z}A,_{p}^2 +H,_{z}A^2 &-2b(z)A,_{z}A,_{p}+K,_{z}\rme^{2 \kappa p}\Sigma^2 \\ \nonumber & \qquad + 2K\rme^{2 \kappa p}\Sigma \Sigma,_{z} + 2A,_{z}\rme^{\kappa p} \Sigma \ \ dp. 
\end{eqnarray*}
We now use the Cauchy-Schwarz inequality, which states that 

\begin{equation*}
\int_{\mathbb{R}} 2 \rme^{\kappa p} \Sigma A,_{z} dp \geq - \int_{\mathbb{R}} \rme^{2 \kappa p} \Sigma^2 + A,_{z}^2 \ dp, 
\end{equation*}
to produce

\begin{eqnarray*}
\fl \frac{dE_{2}}{dz}\geq  \int_{\mathbb{R}} (\beta,_{z}-2a(z)-1)A,_{z}^2-\gamma,_{z}A,_{p}^2&+H,_{z}A^2 -2b(z)A,_{z}A,_{p} \\ \nonumber &+(K,_{z}\rme^{2 \kappa p} -\rme^{2 \kappa p})\Sigma^2 + 2K\rme^{2 \kappa p}\Sigma \Sigma,_{z} dp. 
\end{eqnarray*}
We now wish to deal with the term containing $\Sigma,_{z}$. To do this, we will need the equation of motion for matter, (\ref{udot}). Using this equation, (\ref{rhoandm}) and (\ref{Sigmadef}), it is possible to show that $\Sigma$ is a separable function of $z$ and $r$, i.e. that $\Sigma(z,p)=B(z)C(r)$, where $B(z)=-16 \pi \rme^{-\nu/2}S^2q(z)$ and $C(r)=(U,_{r}-2U)/r^3$. We could therefore write $\Sigma,_{z}=(B,_{z}/B(z)) \Sigma$. Incorporating this produces

\begin{eqnarray}
\label{Idef}
\fl \frac{dE_{2}}{dz} \geq   \int_{\mathbb{R}} (\beta,_{z}-2a(z)-1)A,_{z}^2-\gamma,_{z}A,_{p}^2&+H,_{z}A^2 -2b(z)A,_{z}A,_{p} \\ \nonumber
&+\left( K,_{z}+2K\frac{B,_{z}}{B(z)}\right)\rme^{2 \kappa p} \Sigma^2 dp. 
\end{eqnarray}
Now set $I$ to equal the integrand on the right hand side of (\ref{Idef}) and define $I_{R}=I+ \mu I_{E_{2}}$, where $\mu >0$ is a positive constant, and $I_{E_{2}}$ is the integrand such that $E_{2}(z)=\int_{\mathbb{R}} I_{E_{2}} dp $. If we can show that $I_{R} \geq 0$, then this corollary is proven. We have
\begin{eqnarray*}
\fl I_{R}=(\beta,_{z}-2a(z)-1&+\mu \beta) A,_{z}^2+(-\gamma,_{z}-\mu \gamma)A,_{p}^2+(H,_{z}+\mu H)A^2 \\ \nonumber
&+\left( K,_{z} -1+2K\frac{B,_{z}}{B(z)} + \mu K\right) \rme^{2 \kappa p} \Sigma^2 -2b(z)A,_{z}A,_{p}. 
\end{eqnarray*}
It is possible to pick $K(z)$ so that the $\Sigma^2$ coefficient is always positive so we make this choice. Although $H,_{z}$ is negative, $H(z)$ is positive, and therefore with a choice of large enough $\mu$, the $A^2$ coefficient will also be positive. This leaves us with:
\begin{eqnarray*}
I_{R}& \geq & (\beta,_{z}-2a(z)-1+\mu \beta) A,_{z}^2+(-\gamma,_{z}-\mu \gamma)A,_{p}^2 -2b(z)A,_{z}A,_{p} \\ \nonumber &:=&d(z)A,_{z}^2+e(z)A,_{p}^2+f(z)A,_{z}A,_{p}. 
\end{eqnarray*}
We define the quadratic form
\begin{equation} \label{quadform}
Q(z, p):=d(z) X^2+e(z)Y^2+f(z)XY. 
\end{equation}
In order for this form to be positive definite, we will require $d(z) >0$, $e(z)>0$ and $D(z)=4d(z)e(z)-f(z)^2>0$. We first investigate the behaviour of $d(z)$, $e(z)$ and $f(z)$ on the Cauchy horizon. Using the fact that the Cauchy horizon occurs at $z_{c}=-\rme^{\nu(z_{c})/2}$ we can compute 
$d(z_{c})$ (recall $\beta(z_{c})=0$):

\begin{equation*}
d(z_{c})=  -\rme^{-\nu/2}(-2+\rme^{\nu/2}(1+\dot{\nu})-4(2-\kappa)) |_{z=z_{c}}
\end{equation*}
and so $d(z_{c})$ will be positive so long as $\kappa$ is in the region

\begin{equation*}
\kappa < \frac{1}{4}(10-\rme^{\nu/2}(1+\dot{\nu})):=\tilde{\kappa}(z). 
\end{equation*}
We wish to minimize the function $\tilde{\kappa}(z)$ over the interval $[z_{c}, z_{i}]$. To do so, we first find an explicit functional form for this expression, using (\ref{expform}). We find
\begin{eqnarray*}
\tilde{\kappa}(z)=\frac{1}{4} \left( 10-\frac{9+12az+a^2z(4+3z)}{9(1+az)^{4/3}} \right). 
\end{eqnarray*}
We first set $a=a^*$ to minimize this function with respect to $a$. We can then calculate the derivative of $\kappa^*$ and find that 
\begin{eqnarray*}
\frac{d\tilde{\kappa}}{dz}=\frac{-\rme^{\nu/2}}{4} \left( \frac{2a^2(6+(3-2a)z+3az^2)}{9(1+az)^2(3+az)} \right). 
\end{eqnarray*}
The term in brackets can easily be shown to be positive. Since the coefficient of this bracket above is negative, it follows that the derivative $\frac{d\tilde{\kappa}}{dz}$ is everywhere negative in the region $(z_{c}, z_{i}]$. It follows that the minimal value of $\tilde{\kappa}$ is the value at $z_{i}=0$. Inserting this value produces $\kappa^*=\frac{9}{4}$, which was used in the statements of this corollary and Lemma 2. 

Now, $e(z_c)$ is positive, for any choice of positive $\mu$. This follows since by (\ref{mastercoeff2}), $e(z)=\rme^{-\nu(z)} (\mu - \dot{\nu}(z))$ and using (\ref{expform}) we can check that $\dot{\nu}(z)$ is negative at $z=z_{c}$ for all values of $a \in (0, a^{*})$. Finally, we must check that $D(z_{c})=4d(z)e(z)-b(z)^2>0$. But $d(z_{c})>0$ (with the above choice for $\kappa$) and since $e(z_{c})$ can be made arbitrarily large by a choice of large $\mu$, it follows that $D(z_{c})$ can always be made positive by a suitable choice of $\mu$. Then at the Cauchy horizon, the quadratic form is positive definite. But for a choice of $z^*$ close enough to $z_{c}$, the continuity of the coefficients $d(z)$, $e(z)$ and $f(z)$ ensures that the quadratic form (\ref{quadform}) is positive definite in the range $(z_{c}, z^*]$. Therefore, we can conclude that 

\begin{equation*}
\frac{dE_{2}}{dz} \geq -\mu E_{2}(z)
\end{equation*}
for $z \in (z_{c}, z^*]$. 
\hfill$\square$
\end{proof}
Having successfully bound the derivative of $E_{2}(z)$, we can establish a satisfactory bound on $E_{2}(z)$ itself, which does not share the defects of (\ref{defE1}). 


\begin{theorem} 
\label{Thm2}
Let $A(z,p)$ be a solution to (\ref{master}) which is subject to Theorem \ref{Thm1} and Lemma \ref{Lem1}. Then the energy $E_{2}(z)$ of $A(z,p)$ obeys the \textit{a priori} bound 

\begin{equation*} 
E_{2}(z) \leq C_{1} E_{1}[A](0) + C_{2} J_{\kappa}[\Sigma(z_{i})], 
\end{equation*}
where $J_{\kappa}[\Sigma(z_{i})]=\int_{R} \rme^{2 \kappa p} \Sigma^2(z_{i}, p) dp$ and $z \in (z_{c}, z_{i}]$. 
\end{theorem}
\begin{proof} We can immediately construct a bound on $E_{2}$ by considering the results (\ref{e1abound} - \ref{e1azbound}) of Corollary 2. Using these results, we can construct the bound

\begin{equation} \label{e2firstbound}
E_{2}(z) \leq h(z)\left( E_{1}[A](0) + \int_{-\infty}^{\infty} \Vert \vec{j} \Vert^2 dp \right), 
\end{equation}
where $h(z)=|C_{1}\beta(z)\rme^{C_{0}\bar{z}}|+|\rme^{B_{0}\bar{z}}(H(z)-\gamma(z))|$. The function $h(z)$ clearly diverges on the Cauchy horizon. We now wish to convert the $L^2$-norm of $\vec{j}$ into an \textit{a priori} bound, that is, a bound which depends on some quantity evaluated on the initial data surface. To do this we note that 

\begin{equation*}
\int_{-\infty}^{\infty} \Vert \vec{j} \Vert^2 dp=f(z) J_{\kappa}[\Sigma(z_{i})],
\end{equation*}
where $J_{\kappa}[\Sigma(z_{i})]=\int_{R} \rme^{2 \kappa p} \Sigma^2(z_{i}, p) dp$ and $f(z)=B^{-2}(z_{i})(2B^2(z)\beta^2(z)k^2(z))$ and $k(z)=-\frac{1}{2}(1+z^2\rme^{-\nu})^{-1/2}\rme^{\nu/2}$. By inspection, we can see that the function $f(z)$ is finite up to the Cauchy horizon, so we have the bound 

\begin{equation*}
\int_{0}^{\infty} \Vert \vec{j} \Vert^2 dp \leq C_{0} J_{\kappa}[\Sigma(z_{i})]
\end{equation*}
for some positive and sufficiently large constant $C_{0}$ that depends only on the metric functions. Using this in (\ref{e2firstbound}) produces 

\begin{equation*}
E_{2}(z) \leq h(z)\left( E_{1}[A](0) + C_{0} J_{\kappa}[\Sigma(z_{i})] \right). 
\end{equation*}
We now integrate the bound on $dE_{2}/dz$ from Corollary \ref{Cor3} to find

\begin{equation*}
E_{2}(z) \leq \rme^{-\mu(z-z^*)}E_{2}(z^*)
\end{equation*}
in the range $z \in (z_{c}, z^*]$. Combining these two bounds and noting that $h(z^*)$ is finite results in an \textit{a priori} bound on $E_{2}(z)$:

\begin{equation*} 
E_{2}(z) \leq C_{1} E_{1}[A](0) + C_{2} J_{\kappa}[\Sigma(z_{i})],  
\end{equation*}
where $C_{1}=\rme^{-\mu(z_{c}-z^*)} h(z^*)$ and $C_{2}=\rme^{-\mu(z_{c}-z^*)} h(z^*)C_{0}$ are finite and $z \in (z_{c}, z^*]$. 
\hfill$\square$ 
\end{proof}
Having found an \textit{a priori} bound on $E_{2}(z)$ we can immediately progress to a bound on the function $A(z,p)$. We pause briefly to note that the Sobolev space $\mathbb{H}^{1,2}(\mathbb{R}, \mathbb{R})$ is the set of all functions $f$ with finite $\mathbb{H}^{1,2}$-norm, that is, the set of all functions $f$ such that 
\begin{equation*}
\int_{\mathbb{R}} |f|^2+|f,_{p}|^2 \, dp < \infty. 
\end{equation*}


\begin{theorem}
\label{Thm3}
Let $A(z,p)$ be a solution to (\ref{master}) which is subject to Theorem \ref{Thm1} and Lemma \ref{Lem1}. Then $A(z,p)$ is uniformly bounded on $(z_{c}, z_{i}]$. That is, there exists constants $C_{1}>0$, $C_{2}>0$ such that 
\begin{equation*} 
\vert A(z, p) \vert \leq C_{1}E_{1}[A](0)+C_{2}J_{\kappa}[\Sigma(z_{i})]. 
\end{equation*}
\end{theorem}
\begin{proof} From the previous theorem and the fact that (in the definition of $E_{2}(z)$) the terms $\beta A,_{z}^2+K(z)\rme^{2 \kappa p} \Sigma^2$ are positive definite, we can state that 

\begin{equation*}
\int_{\mathbb{R}} A^2+A,_{p}^2 dp \leq C_{1}E_{1}(0)+C_{2}J_{\kappa}[\Sigma(z_{i})]. 
\end{equation*}
So we get a bound on the $\mathbb{H}^{1,2}(\mathbb{R}, \mathbb{R})$ norm of $A(z,p)$ directly from Theorem \ref{Thm2}. We now apply Sobolev's inequality, 

\begin{equation*}
\vert A \vert \leq \frac{1}{2} \int_{\mathbb{R}} \vert A \vert^2 + \vert A,_{p} \vert^2 dp,
\end{equation*}
to convert this to a bound on $A(z,p)$:

\begin{equation*}
\vert A(z, p) \vert^2 \leq C_{1}E_{1}(0)+C_{2}J_{\kappa}[\Sigma(z_{i})]
\end{equation*}
for all $z \in (z_{c}, z_{i}]$. 
\hfill$\square$
\end{proof}
\textbf{Remark 4.1:}  This theorem shows that $A(z,p)$ (and therefore the gauge invariant matter scalar $\Psi$) is bounded in the approach to the Cauchy horizon. However, this is not itself sufficient to prove that the limit of $\Psi$ (for all $p \in \mathbb{R}$) actually exists in the approach to the Cauchy horizon. The following lemma allows us to control the behaviour of the time derivative of $A(z,p)$ and hence, to prove the existence and finiteness of the limit. 


\begin{lemma}
\label{Lem2}
Let $A(z,p)$ be a solution to (\ref{master}) which is subject to Theorem \ref{Thm1} and Lemma \ref{Lem1}. Then $A,_{z}(z,p)$ is uniformly bounded on $(z_{c}, z_{i}]$. That is, there exist constants $\{C_{i} \}, \ i=0,..,5$ such that 
\begin{eqnarray} \label{azbound}
\fl \vert A,_{z}(z,p) \vert \leq C_{0}E_{1}[A](0)+&C_{1}E_{1}[A,_{p}](0)+C_{2}E_{1}[A,_{pp}](0)+C_{3}J_{\kappa}[\Sigma(z_{i})] \\ \nonumber & \qquad \qquad \qquad \qquad +C_{4}J_{\kappa}[\Sigma,_{p}(z_{i})]+ C_{5}J_{\kappa}[\Sigma,_{pp}(z_{i})]. 
\end{eqnarray}
\end{lemma}
\begin{proof} We wish to find a bound on the behaviour of $A,_{z}(z,p)$. To achieve this, we first rewrite (\ref{master}) as a first order transport equation for $A,_{z}(z,p)$. If we label $\chi:=A,_{z}$ then

\begin{equation} \label{transport}
\beta(z) \chi,_{z}+ \xi(z) \chi,_{p} + a \chi=f(z,p), 
\end{equation}
where $f(z,p)=\rme^{ \kappa p } \Sigma -c(z)A(z,p) -b(z) A,_{p} - \gamma(z) A,_{pp}$. By inspection, we see that the function $f(z, p)$ is smooth and has compact support on each $z=constant$ surface. If we define the differential operator $L$ to be

\begin{equation*}
L:=\beta(z) \frac{d^2}{dz^2} + \gamma(z) \frac{d^2}{dp^2} + \xi(z) \frac{d^2}{dzdp} + a(z) \frac{d}{dz}+b(z) \frac{d}{dp} +c(z)
\end{equation*}
then (\ref{master}) would read

\begin{equation*}
L[A]=a_{0}(p) \Sigma(z,p),
\end{equation*}
where $a_{0}(p)=\rme^{\kappa p}$. Now since every coefficient in the above differential operator has only $z$-dependence, we could differentiate (\ref{master}) with respect to $p$ and write the result as 

\begin{equation*}
L[A,_{p}]=b_{0}(p) \Sigma + a_{0}(p) \Sigma,_{p}, 
\end{equation*}
where $b_{0}(p)=da_{0}/dp=\kappa a_{0}(p)$. Similarly, 

\begin{equation*}
L[A,_{pp}]=c_{0}(p) \Sigma + 2 b_{0}(p) \Sigma,_{p} + a_{0}(p) \Sigma,_{pp}, 
\end{equation*}
where $c_{0}(p)=db_{0}/dp= \kappa^2 a_{0}(p)$. So we see that $A,_{p}$ and $A,_{pp}$ satisfy similar differential equations to $A(z,p)$, with different source terms. We can therefore apply Theorem \ref{Thm3} to $A,_{p}$ and $A,_{pp}$ so long as we modify the bounding terms to take account of the modified source terms:

\begin{equation} \label{apbound}
\vert A,_{p} \vert \leq C_{3} J_{\kappa }[\Sigma(z_{i})]+C_{4} J_{\kappa }[\Sigma,_{p}(z_{i})]+C_{5}E[A,_{p}](0)
\end{equation}

\begin{equation} \label{appbound}
\vert A,_{pp} \vert \leq C_{6} J_{\kappa }[\Sigma(z_{i})]+C_{7} J_{\kappa }[\Sigma,_{p}(z_{i})]+C_{8}J_{\kappa}[\Sigma,_{pp}(z_{i})]+C_{9}E[A,_{pp}](0)
\end{equation}
We must now integrate the first order transport equation (\ref{transport}) and use the above results to bound $A,_{z}(z,p)$. The charateristics of (\ref{transport}) are $dp/dz=\xi(z)/\beta(z)$, which integrates to give

\begin{equation*} 
p=\alpha+\int_{z}^{z_{i}} \frac{\xi(s)}{\beta(s)} ds =\alpha + \omega(z). 
\end{equation*}
$\alpha$ labels each characteristic, and at $z=z_{i}$, it gives the value of $p$ where the characteristic intersects the initial data surface. With this result, the transport equation becomes 

\begin{equation*}
\beta(z)\frac{d}{dz} \{ \chi(z, \alpha + \omega(z)) \} + a(z) \chi(z, \alpha + \omega(z)) = f(z, \alpha + \omega(z)),  
\end{equation*}
where the derivative is taken along characteristics. Now define 

\begin{equation*}
J(z):=\exp \left[ \int_{z}^{z_{i}} \frac{a(s)}{\beta(s)} ds  \right]. 
\end{equation*}
It can easily be verified that a solution of the ordinary differential equation (\ref{transport}) can be written as 

\begin{equation} \label{azsoln}
J(z) \chi(z, \alpha + \omega(z))=\chi(z_{i}, \alpha)+ \int_{z}^{z_{i}} \frac{J(s)}{\beta(s)} f(s, \alpha + \omega(s)) ds.  
\end{equation}
Now recall that $\omega(z)=\int_{z}^{z_{i}} \frac{\xi(s)}{\beta(s)} ds $ which tends to infinity as $z \rightarrow z_{c}$. For $z$ close enough to $z_{c}$ the characteristic at $(z,p)$ will hit $z=z_{i}$ at some very large negative $p$ value. Therefore, since $A(z,p)$ has compact support, $\chi(z_{i}, \alpha)=0$. We now apply the mean value theorem for integrals to find that 
\begin{equation*}
\int_{z}^{z_{i}} \frac{J(s)}{\beta(s)} f(s, \alpha + \omega(s)) ds = f(z_{*}, \alpha + \omega(z_{*})) \int_{z}^{z_{i}} \frac{J(s)}{\beta(s)} ds 
\end{equation*}
for some $z_{*}$ in the interval $(z_{c}, z_{i}]$. Then from (\ref{azsoln}) we can conclude that

\begin{equation*}
\chi(z, \alpha + \omega(z)) = \frac{f(z_{*}, \alpha + \omega(z_{*}))}{J(z)} \int_{z}^{z_{i}} \frac{J(s)}{\beta(s)} ds.  
\end{equation*}
The coefficient of $f(z_{*}, \alpha + \omega(z_{*}))$ above is clearly finite away from the Cauchy horizon, and could therefore be bounded by some suitably large constant $C_{*}$, so 

\begin{equation*}
\chi(z, \alpha + \omega(z)) \leq f(z_{*}, \alpha + \omega(z_{*}))C_{*}. 
\end{equation*}
Now from (\ref{transport}), (\ref{apbound}) and (\ref{appbound}), we know that $f(z,p)$ is bounded, so we may finally state the bound on $\chi$:
\begin{eqnarray*} 
\fl \vert \chi(z,p) \vert:=\vert A,_{z}(z,p) \vert \leq &C_{0}E_{1}[A](0)+C_{1}E_{1}[A,_{p}](0)+C_{2}E_{1}[A,_{pp}](0) \\ \nonumber & \qquad +C_{3}J_{\kappa}[\Sigma(z_{i})]+C_{4}J_{\kappa}[\Sigma,_{p}(z_{i})]+ C_{5}J_{\kappa}[\Sigma,_{pp}(z_{i})]
\end{eqnarray*}
in the range $z \in (z_{c}, z_{i}]$. 
\hfill$\square$
\end{proof}
Having bounded the derivative of $A(z,p)$, we are now in a position to bound the perturbation on the Cauchy horizon. 


\begin{theorem}
\label{Thm4}
Let $A(z,p)$ be a solution of (\ref{master}) subject to Theorem \ref{Thm1} and Lemma \ref{Lem1}. Then $A_{\mathcal{H}+}:=\lim_{z \rightarrow z_{c}} A(z, \cdot) \in C^{\infty}(\mathbb{R}, \mathbb{R})$ obeys the bound 

\begin{equation*} 
\vert A_{\mathcal{H}+}(z, p) \vert \leq C_{1}E_{1}[A](0) + C_{2}J_{\kappa }[\Sigma(z_{i})]. 
\end{equation*}
\end{theorem}
\begin{proof} We wish to show that $\lim_{z \rightarrow z_{c}} \vert A(z,p) \vert$ is bounded. We begin by fixing $p$ and introducing a sequence of $z$-values that converge to $z_{c}$, $\{z^{(n)}\}_{n=0}^{\infty} \subset (z_{c}, z_{i}]$. For all $m,n \geq 1$, we can use the mean value theorem to show that  

\begin{equation*}
\vert A(z^{(m)}, p) - A(z^{(n)}, p) \vert = \vert A,_{z}(z_{*}, p) \vert \vert z^{(m)}-z^{(n)} \vert
\end{equation*}
for some $z_{*} \in (z^{(m)}, z^{(n)})$. Then Lemma \ref{Lem2} tells us that $A,_{z}$ is bounded, so $\vert A,_{z}(z_{*}, p) \vert$ will be a real number. Then since the sequence $\{z^{(n)}\}_{n=0}^{\infty} \subset (z_{c}, z_{i}]$ tends towards $z_{c}$, it follows that for large enough $n,m$, $\vert z^{(m)}-z^{(n)} \vert < \epsilon$ for all $\epsilon >0$. Therefore $A(z^{(m)}, p)$ is a Cauchy sequence of real numbers. Then for each $p \in \mathbb{R}$, $\lim_{z \rightarrow z_{c}} A(z,p)$ exists. Define

\begin{equation*}
A_{\mathcal{H}+}:=\lim_{z \rightarrow z_{c}} A(z, \cdot). 
\end{equation*}
We now wish to take the limit $z \rightarrow z_{c}$ in Theorem \ref{Thm3}, which bounds $\vert A(z,p) \vert$. In order to do this, we will need to know that the limits of $A,_{p}$ and $A,_{pp}$ exist. But this follows by a similar argument to the above (recall that we know that all $p$- derivatives of $A(z,p)$ to arbitrary order can be bounded, using an argument similar to that of Lemma \ref{Lem2}). Finally, we must show that 
\begin{equation*}
\frac{d}{dp} A_{\mathcal{H}+} = \lim_{z \rightarrow z_{c}} A,_{p}. 
\end{equation*}
But we know that the sequence $A(z^{(n)}, p)$ converges uniformly to $A(z,p)$, so the above result follows. Using these results, we can take the limit $z \rightarrow z_{c}$ in Theorem \ref{Thm3} to find that 

\begin{equation*}
\vert A_{\mathcal{H}+}(z, p) \vert \leq C_{1}E_{1}[A](0) + C_{2}J_{\kappa }[\Sigma(z_{i})]. 
\end{equation*}
\hfill$\square$
\end{proof}
We now wish to generalize our choice of initial data. Recall that we chose an initial data surface which intersected the axis at $r=0$. We therefore had to require that the initial data for the perturbation be supported away from this point, which is an undesirable feature of our analysis so far. We pause briefly to note that the Sobolev spaces $\mathbb{H}^{2,2}(\mathbb{R}, \mathbb{R})$ and $\mathbb{H}^{3,2}(\mathbb{R}, \mathbb{R})$ are the set of all functions $f$ with finite $\mathbb{H}^{2,2}$ and finite $\mathbb{H}^{3,2}$-norms respectively, that is, the set of all functions $f$ such that 
\begin{equation*}
\int_{\mathbb{R}} |f|^2+|f,_{p}|^2 + |f,_{pp}|^2 \, dp < \infty,
\end{equation*}
for $f \in \mathbb{H}^{2,2}$ and 
\begin{equation*}
\int_{\mathbb{R}} |f|^2+|f,_{p}|^2 + |f,_{pp}|^2 + |f,_{ppp}|^2 \, dp < \infty,
\end{equation*}
for $f \in \mathbb{H}^{3,2}$ respectively. 

\begin{theorem}
\label{Thm5}
Let $\kappa \in [0, \kappa^*)$. 
\\ \\ 
(1) Let $f \in \mathbb{H}^{1,2}(\mathbb{R}, \mathbb{R})$, $g \in L^2(\mathbb{R}, \mathbb{R})$ and $\Sigma \in L^2(\mathbb{R}, \mathbb{R})$ for each fixed $z$. Then there exists a unique solution $A \in C((z_{c}, z_{i}], \mathbb{H}^{1,2}(\mathbb{R}))$ of the initial value problem consisting of (\ref{master}) with the initial data $A |_{z_{i}}=f$, $A,_{z} |_{z_{i}}=g$. 
This solution satisfies the a priori bound 

\begin{equation*} 
\vert A(z,p) \vert \leq C_{0} E_{1}[A](0) + C_{2} J_{\kappa}[\Sigma(z_{i})]
\end{equation*}
for $z \in (z_{c}, z_{i}]$ and $p \in \mathbb{R}$. 
\\ \\
(2) Let $f \in \mathbb{H}^{3,2}(\mathbb{R}, \mathbb{R})$, $g \in \mathbb{H}^{2,2}(\mathbb{R}, \mathbb{R})$ and $\Sigma \in \mathbb{H}^{2,2}(\mathbb{R}, \mathbb{R})$ for each fixed $z$. Then there exists a unique solution $A \in C([z_{c}, z_{i}], \mathbb{H}^{1,2}(\mathbb{R}))$ of the initial value problem consisting of (\ref{master}) with the initial data $A |_{z_{i}}=f$, $A,_{z} |_{z_{i}}=g$. This solution satisfies the \textit{a priori} bound 
\begin{equation*} 
\vert A(z,p) \vert \leq C_{0} E_{1}[A](0) + C_{2} J_{\kappa}[\Sigma(z_{i})]
\end{equation*}
for $z \in (z_{c}, z_{i}]$ and $p \in \mathbb{R}$, and its time derivative satisfies 
\begin{eqnarray*} 
\fl \vert A,_{z}(z,p) \vert \leq C_{0}E_{1}[A](0) + C_{1}E_{1}[A,_{p}](0)+&C_{2}E_{1}[A,_{pp}](0) + C_{3}J_{\kappa}[\Sigma(z_{i})] \\ \nonumber 
& \qquad +C_{4}J_{\kappa}[\Sigma,_{p}(z_{i})] + C_{5}J_{\kappa}[\Sigma,_{pp}(z_{i})]
\end{eqnarray*}
for $z \in (z_{c}, z_{i}]$ and $p \in \mathbb{R}$. 
\end{theorem}

\begin{proof} The proof of this theorem is standard, and uses the density of the space $C_{0}^{\infty}(\mathbb{R}, \mathbb{R})$ in the Banach spaces $\mathbb{H}^{1,2}(\mathbb{R}, \mathbb{R})$, $\mathbb{H}^{2,2}(\mathbb{R},\mathbb{R})$, $\mathbb{H}^{3,2}(\mathbb{R},\mathbb{R})$ and $L^2(\mathbb{R}, \mathbb{R})$. The proof is essentially identical to that of theorem 5 in \cite{vaidya}. 

\hfill$\square$
\end{proof}
\textbf{Remark 4.2:} The choice of which Sobolev space to take our initial data functions from in the above proofs is dictated by the nature of the bounds required. For example, to use a bound involving $E_{1}[A,_{pp}]$, we will require the function $f$ to be in $\mathbb{H}^{3,2}(\mathbb{R})$ so that it and its $p$-derivatives up to third order are in $L^2(\mathbb{R})$. This is required for the integral involved in $E_{1}[A,_{pp}]$ to be well defined. All other choices of Sobolev spaces used above can be understood in a similar fashion. 
\\ \\ 
\textbf{Remark 4.3} This theorem successfully generalizes the choice of initial data function for (\ref{master}). This generalisation involves choosing initial data which need not vanish at the scaling centre of the spacetime, which is crucial as it allows for a perturbation which need not vanish at the past endpoint of the naked singularity. Similar finiteness results go through for this general choice of initial data. 


\section{Physical Interpretation of Results}
\label{sec:weyl}

In order to physically interpret the results obtained thus far, we turn to the perturbed Weyl scalars. These scalars are related to the gauge invariant scalar $\Psi$ and can be interpreted in terms of in- and outgoing gravitational radiation. In the case of odd parity perturbations, they are both tetrad and identification gauge invariant. This means that if we make a change of null tetrad, or a change of our background coordinate system, we will find that these terms are invariant under such changes. 

Following \cite{physical} and \cite{szekeres}, we note that $\delta \Psi_{0}$ and $\delta \Psi_{4}$ represent transverse gravitational waves propagating radially inwards and outwards, and $\delta \Psi_{2}$ represents the perturbation of the Coulomb part of the gravitational field\footnotemark.

\footnotetext{We note that \cite{szekeres} refers to $\delta \Psi_{1}$ and $\delta \Psi_{3}$ as ``longitudinal gravitational waves'' propagating radially inwards and outwards. }

The perturbed Weyl scalars are given by
\begin{eqnarray*}
\delta \Psi_{0} = \frac{Q_{0}}{2 R^2} \bar{l}^{A} \bar{l}^{B} k_{A|B}, \\ 
\delta \Psi_{1}=\frac{Q_{1}}{R} \left( (R^2 \Psi)_{|A} \bar{l}^{A} - \frac{4}{R^2} k_{A} \bar{l}^A \right), \\
\delta \Psi_{2}=Q_{2} \Psi, \\ 
\delta \Psi_{3}=\frac{Q^*_{1}}{R} \left( (R^2 \Psi)_{|A} \bar{n}^A - \frac{4}{R^2} k_{A} \bar{n}^A \right), \\ 
\delta \Psi_{4}=\frac{Q_{0}^*}{2R^2} \bar{n}^A \bar{n}^B k_{A|B},
\end{eqnarray*}
where $\Psi$ is the gauge invariant scalar appearing in (\ref{master}), $k_{A}$ is the gauge invariant vector describing the metric perturbation, (\ref{ka}), and $\bar{l}^{A}$ and $\bar{n}^{A}$ are the in- and outgoing null vectors given in (\ref{inoutvectors}). $Q_{0}$, $Q_{1}$ and $Q_{2}$ are angular coefficients depending on the other vectors in the null tetrad, and on the basis constructed from the spherical harmonics. We have made a gauge choice such that the perturbation of the real members of the null tetrad vanishes, that is, $\delta l_{\mu} = \delta n_{\mu} =0$. See \cite{physical} for further details. 

We note that the quantities $\delta P_{-1}$, $\delta P_{0}$ and $\delta P_{+1}$, which are defined as follows, 
\begin{equation}
\label{deltaPs1}
\delta P_{-1} = |\delta \Psi_{0} \delta \Psi_{4}|^{1/2}, 
\end{equation}
\begin{equation}
\label{deltaPs2}
\delta P_{0} = \delta \Psi_{2}, 
\end{equation}
\begin{equation} 
\label{deltaPs3}
\delta P_{+1} = |\delta \Psi_{1} \delta \Psi_{3}|^{1/2},
\end{equation}
are fully gauge invariant (in that they are invariant under a change in the background null tetrad, as well as being invariant under transformations in the perturbed null tetrad and identification gauge transformations) and have physically meaningful magnitudes. 

Although we could write these scalars in terms of the coordinates $(z,p)$ used in the previous section, it is advantageous to use null coordinates $(u,v)$ instead, as this simplifies matters considerably. We will therefore consider the master equation in null coordinates, and establish a series of results indicating the boundedness of various of the derivatives of $A(u,v)$ in null coordinates. These results will allow us to show that the perturbed Weyl scalars are bounded as the Cauchy horizon is approached. 


\subsection{Master Equation in Null Coordinates}
\label{sec:masternull}
We first rewrite the master equation (\ref{master}) in terms of the in and out-going null coordinates (\ref{uvcoords}). The master equation takes the form

\begin{equation} \label{masteruv}
\alpha_{1}(u,v) \, A,_{uv} + \alpha_{2}(u,v) u \, A,_{u} + \alpha_{3}(u,v) v \, A,_{v} + \alpha_{4}(u,v) A = \rme^{\kappa p} \Sigma(u,v), 
\end{equation}
where in terms of the coefficients (\ref{mastercoeff1} -\ref{mastercoeff6}), the above coefficients are given by  

\begin{eqnarray} 
\label{uvcoefficients}
\alpha_{1}(u,v)=2z \frac{\beta(z) + \xi(z)}{f_{+}(z) f_{-}(z)} + 2\gamma(z), \qquad \quad \quad
\alpha_{2}(u,v)&=\frac{a(z)}{f_{+}(z)} + b(z), \\ 
\alpha_{3}(u,v)=\frac{a(z)}{f_{-}(z)} + b(z), \qquad \qquad \qquad \qquad
\alpha_{4}(u,v)&=c(z),
\end{eqnarray}
where $f_{\pm}(z)$ are factors coming from (\ref{uvcoords}). We can formally solve (\ref{masteruv}) by integrating across the characteristic diamond $\Omega=\{(\bar{u},\bar{v}) : u_{0} < \bar{u} \leq u, v_{o} \leq \bar{v} \leq v \}$ (see figure \ref{Fig3}). We find 

\begin{equation*} 
A(u,v)=A(u_{0}, v) + A(u, v_{0}) + A(u_{0}, v_{0}) + \int_{u_{0}}^u \int_{v_{0}}^v \, F(\bar{u}, \bar{v}) d \bar{u} d\bar{v}, 
\end{equation*}
where $F(u,v)=(\alpha_{1})^{-1} \left(-\alpha_{2}(u,v) u \, A,_{u} - \alpha_{3}(u,v) v \, A,_{v} - \alpha_{4}(u,v) A + \rme^{\kappa p} \Sigma(u,v) \right)$.


\begin{figure}
\begin{center}
\includegraphics[scale=1.3]{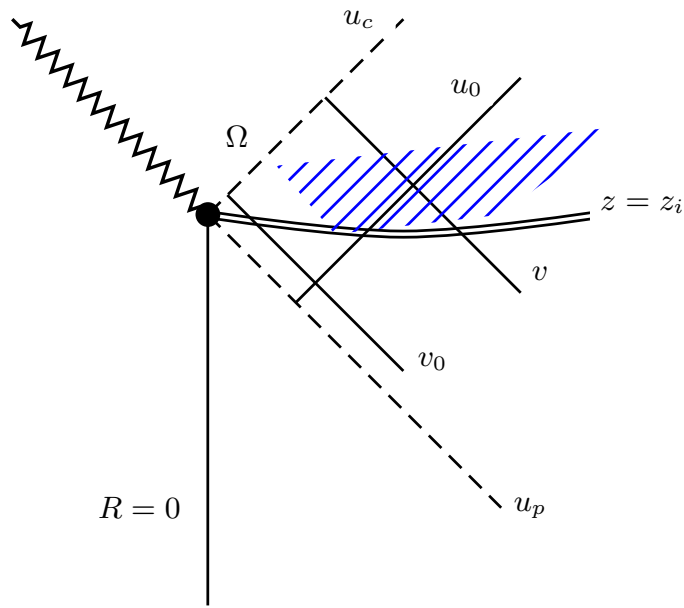}
\end{center}
\caption{The characteristic diamond. We integrate over the characteristic diamond labelled $\Omega$, where $u$ and $v$ are the retarded and advanced null coordinates. $u_{c}$ labels the Cauchy horizon, $u_{p}$ labels the past null cone of the naked singularity and $z_{i}$ is the initial data surface. 
}
\label{Fig3}
\end{figure}

Now in section \ref{sec:giinterp}, in order to control the perturbed Weyl scalars, we will need to know that $A$, $A,_{u}$, $A,_{v}$, $A,_{uu}$ and $A,_{vv}$ are bounded in the approach to the Cauchy horizon. 


\begin{lemma}
\label{Lem3}
With a choice of initial data $A(z_{i}, p) = f(p)$, $A,_{z}(z_{i}, p) = j(p)$ and $\Sigma(z_{i}, p) = h(p)$, with $f(p)$, $j(p)$ and $h(p) \in C_{0}^{\infty}(\mathbb{R}, \mathbb{R})$, the first order derivatives of $A(u,v)$ with respect to $u$ and $v$ are bounded by \textit{a priori} terms in the approach to the Cauchy horizon. 
\end{lemma}
See appendix C for the proof of this lemma. We can use this result, together with results from section \ref{sec:e2}, to establish that the second order derivatives of $A$ with respect to $u$ and $v$ are also bounded. 


\begin{lemma}
\label{Lem4}
With a choice of initial data $A(z_{i}, p) = f(p)$, $A,_{z}(z_{i}, p) = j(p)$ and $\Sigma(z_{i}, p) = h(p)$, with $f(p)$, $j(p)$ and $h(p) \in C_{0}^{\infty}(\mathbb{R}, \mathbb{R})$, the second order derivatives $A,_{vv}$, $A,_{uu}$ and $A,_{uv}$ of $A$ are bounded by \textit{a priori} terms in the approach to the Cauchy horizon. 
\end{lemma}
See appendix C for the proof of this lemma. The results so far establish the boundedness of all first and second order derivatives of $A$ with respect to $u$ and $v$, with a choice of initial data from the space $C_{0}^{\infty}(\mathbb{R}, \mathbb{R})$. As discussed in section \ref{sec:chbehaviour}, this choice of initial data does not interact with the past endpoint of the naked singularity. As in Theorem \ref{Thm5}, we can extend this choice of initial data so that the perturbation need not vanish at the Cauchy horizon. 

\begin{lemma}
\label{Lem5}
With a choice of initial data $A(z_{i}, p) = f(p)$, $A,_{z}(z_{i}, p) = j(p)$ and $\Sigma(z_{i}, p) = h(p)$, with $f(p) \in \mathbb{H}^{3,2}(\mathbb{R}, \mathbb{R})$, $j(p) \in \mathbb{H}^{1,2}(\mathbb{R}, \mathbb{R})$ and $h(p) \in \mathbb{H}^{3,2}(\mathbb{R}, \mathbb{R})$, the first and second order derivatives $A,_{u}$, $A,_{v}$, $A,_{vv}$, $A,_{uu}$ and $A,_{uv}$ of $A$ are bounded by \textit{a priori} terms in the approach to the Cauchy horizon. 
\end{lemma}
See appendix C for the proof of this lemma. Having bounded the first and second order derivatives of $A$ with a satisfactory choice of initial data, we are now in a position to consider the perturbed Weyl scalars. 


\subsection{Gauge Invariant Curvature Scalars}
\label{sec:giinterp}
The in- and outgoing background null vectors $\bar{l}^{\mu}$ and $\bar{n}^{\mu}$ are given in (\ref{inoutvectors}). We note that a factor of $B^{-1}(u,v)$ appears in the definition of $\bar{l}^{\mu}$, and that this factor involves a power of $r^{-2}$. 

In $(u,v)$ coordinates, the perturbed Weyl scalars take the form
\begin{equation}
\label{weyl1}
\fl \delta \Psi_{0} = \frac{Q_{0} B^{-2}}{2 \mathcal{L} S^2}\left( 16 \pi  ((S^2 L_{0}),_{u} - \gamma_{0}S^2 L_{0}) + \frac{r^2}{B} \left( (S^4 \Psi),_{uu} - \gamma_{0} (S^4 \Psi),_{u} \right) \right), 
\end{equation}
\begin{equation}
\label{weyl2}
\fl \delta \Psi_{1}=\frac{Q_{1}}{S B} \left( r  (S^2 \Psi),_{u}-\frac{4}{\mathcal{L} B r} \left( 16 \pi S^2 L_{0} -\frac{r^2 (S^4 \Psi),_{u}}{B} \right) \right), 
\end{equation}
\begin{equation}
\label{weyl3}
\fl \delta \Psi_{2}=Q_{2} \Psi, 
\end{equation}
\begin{equation}
\label{weyl4}
\fl \delta \Psi_{3}=\frac{Q_{1}^{*}}{S} \left( r (S^2 \Psi),_{v} - \frac{4}{\mathcal{L}r} \left( 16 \pi S^2 L_{1} - \frac{r^2 (S^4 \Psi),_{v}}{B} \right) \right), 
\end{equation}
\begin{equation}
\label{weyl5}
\fl \delta \Psi_{4}=\frac{Q^{*}_{0}}{2 \mathcal{L} S^2}\left( 16 \pi  ((S^2 L_{0}),_{v} - \gamma_{1}S^2 L_{1}) + \frac{r^2}{B} \left( (S^4 \Psi),_{vv} - \gamma_{1} (S^4 \Psi),_{v} \right) \right),
\end{equation}
where we used (\ref{psistoka}) to write $\delta \Psi_{0}$ and $\delta \Psi_{4}$ in terms of $\Psi$. Here, $\gamma_{0}(u,v)$ and $\gamma_{1}(u,v)$ are Christoffel symbols, $\mathcal{L}=(l-1)(l+2)$ and $L_{A}=(L_{0}, L_{1})$ is the gauge invariant matter vector (\ref{la}). 

\begin{theorem}
\label{Thm6}
With a choice of initial data $\Psi(z_{i}, p) = f(p)$, $\Psi,_{z}(z_{i}, p) = j(p)$ and $\Sigma(z_{i}, p) = h(p)$, with $f(p) \in \mathbb{H}^{3,2}(\mathbb{R}, \mathbb{R})$, $j(p) \in \mathbb{H}^{1,2}(\mathbb{R}, \mathbb{R})$ and $h(p) \in \mathbb{H}^{3,2}(\mathbb{R}, \mathbb{R})$, the perturbed Weyl scalars, as well as $\delta P_{-1}$, $\delta P_{0}$ and $\delta P_{+1}$, remain finite on the Cauchy horizon, barring a possible divergence at the past endpoint of the naked singularity, where $r =0$. They are bounded by \textit{a priori} terms arising from the bounds on $A$, $A,_{z}$, $A,_{p}$, $A,_{pp}$, $A,_{zp}$ and $\Sigma$.
\end{theorem}
\begin{proof}
If we consider (\ref{weyl1} - \ref{weyl5}), we see that the perturbed Weyl scalars depend on the gauge invariant scalar $\Psi$, its first derivatives $\Psi,_{u}$ and $\Psi,_{v}$, its second derivatives $\Psi,_{uu}$ and $\Psi,_{vv}$, and on the gauge invariant vector $L_{A}$. By letting $\kappa =0$ in Lemma \ref{Lem5} we can immediately state that $\Psi$, $\Psi,_{u}$, $\Psi,_{v}$, $\Psi,_{uu}$ and $\Psi_{vv}$ remain finite up to and on the Cauchy horizon. They are bounded by \textit{a priori} terms arising from the bounds on $A$, $A,_{z}$, $A,_{p}$, $A,_{pp}$, $A,_{zp}$ and $\Sigma$. 

Thus, the perturbed Weyl scalars remain finite on the Cauchy horizon, and are bounded by the same \textit{a priori} terms, except for a possible divergence at $r=0$. The terms involving $L_{A}$ depend on the function $U(r)$, and these may also diverge at $r=0$, depending on the details of $U(r)$. 

From (\ref{deltaPs1} - \ref{deltaPs3}), $\delta P_{-1}$, $\delta P_{0}$ and $\delta P_{+1}$ are given by products of the perturbed Weyl scalars, and therefore are bounded in the same way, with a similar proviso about a possible divergence at $r=0$. 
\hfill$\square$
\end{proof}
This theorem establishes that the perturbed Weyl scalars remain finite in the approach to the Cauchy horizon, and we can conclude that the various gravitational waves and the perturbation of the Coulomb potential represented by these scalars also remain finite up to and on the Cauchy horizon. 

Having studied the behaviour of the perturbed Weyl scalars, it is reasonable to ask whether there are any scalars arising from the perturbed Ricci tensor which we should also consider. We are not aware of any gauge invariant scalars which can be constructed from the perturbed Ricci tensor, but we expect that any such scalars would be related via the Einstein equations to gauge invariant matter scalars. In section \ref{sec:matter}, we showed that the matter perturbation depends only on an initial data function, and therefore, we expect any such scalars to be trivial in this sense. 


\section{The $l=1$ Perturbation}
\label{sec:l1pert}
We now consider separately the behaviour of the $l=1$ perturbation. When $l=1$, $k_{A}$ is no longer gauge invariant. Instead, we find that under a change of coordinates $\vec{x} \rightarrow \vec{x}'=\vec{x}+\vec{\xi}$, where $\vec{\xi}=\xi S_{a}dx^{a}$, 
\begin{equation*}
k_{A} \rightarrow k_{A} - r^2 (r^{-2} \xi),_{A}. 
\end{equation*}
Additionally, (\ref{kaein}) no longer holds. 

However, $\Psi$ is still gauge invariant, and obeys (\ref{psistoka}). When $l=1$, $\mathcal{L}=0$, so that (\ref{psistoka}) reduces to
\begin{equation}
\label{psieqn}
16 \pi r^2 L_{A} - \epsilon_{AB} (r^4 \Psi)^{|B}=0. 
\end{equation}
Now, the stress-energy conservation equation, (\ref{secons}) reduces to $(r^2L_{A})^{|A}=0$ when $l=1$. This indicates that there exists a potential for $L_{A}$, which we write as 
\begin{equation}
\label{potential}
r^2L_{A}=\epsilon_{A}^{\,\, B} \lambda,_{B}. 
\end{equation}
As before,  in $(t,r)$ coordinates, $L_{A}=(\bar{\rho}(t,r) U(r), 0)$, so (\ref{potential}) implies that $\lambda'(t,r)=r^2 \bar{\rho}(t,r) U(r)$, where in this section, we revert to the notation $\cdot=\frac{\partial}{\partial t}$ and $'=\frac{\partial}{\partial r}$. Combining (\ref{psieqn}) and (\ref{potential}) produces
\begin{equation*}
\epsilon_{A}^{\,\,B} (16 \pi \lambda -  r^4 \Psi),_{B}=0,
\end{equation*}
which implies that 
\begin{equation*}
r^4 \Psi(t,r)=16 \pi \lambda(t,r) + c,
\end{equation*}
where $c \in \mathbb{R}$ is a constant. This result indicates that $\Psi$ remains finite up to and on the Cauchy horizon, barring a possible divergence at $r = 0$; whether or not this divergence occurs depends on the choice of the initial velocity perturbation $U(r)$. 



\section{Conclusion}
\label{sec:conclusion}

We have considered here odd parity perturbations of the self-similar Lema\^{i}tre-Tolman-Bondi spacetime. More precisely, we have considered the multipoles of the perturbation, that is, the coefficents of the scalar, vector and tensor bases constructed from the spherical harmonics into which the perturbation may be decomposed. We have examined the evolution of a gauge invariant scalar $ \Psi $ which completely determines the perturbation. This scalar acts as a potential for the metric perturbation, the matter perturbation having been fully determined by a specification of initial data. Additionally, $\Psi$ is related to the perturbed Weyl scalars which represent transverse gravitational waves moving radially inwards and outwards along null directions and the perturbation of the Coulomb component of the gravitational field \cite{physical}. 


We have found that the scalar $\Psi$ remains finite as it impinges on the Cauchy horizon of the naked singularity. Finiteness refers to certain natural integral energy measures (as well as pointwise values thereof) which arise in this spacetime, whose value bounds the growth of this scalar. For the analysis of the Cauchy horizon behaviour, we used a foliation of this spacetime which consists of hypersurfaces that are generated by the homothetic Killing vector field. This is a natural choice to make, as it exploits the self-similarity of the background spacetime. If we use this foliation, we find that the coefficients of the master equation are independent of the radial coordinate. This foliation also dictates our choice of initial data surface for the Cauchy problem. 

A disadvantage of this choice of foliation is that these hypersurfaces intersect the singular scaling origin of the spacetime, rather than meeting the regular centre $R=0$. This forced us to begin our analysis by considering initial data taken from the space $C_0^{\infty}(\mathbb{R}, \mathbb{R})$ which were compactly supported away from the singular point. We then established Theorems \ref{Thm1}-\ref{Thm4} using this data. We finally extended these results to a more general choice of initial data, taken from various Sobolev spaces, which were capable of having non-zero values at the singular origin. This extension is crucial, as it shows that a perturbation which interacts with the naked singularity still remains bounded at the Cauchy horizon. 

Using the perturbed Weyl scalars, one can give a physical interpretation of these results; the gauge invariant scalar $\Psi$ enters into the definition of the perturbed Weyl scalars, which in turn represent ingoing and outgoing gravitational radiation and the perturbation of the Coulomb part of the gravitational field. Now since $\Psi$ remains finite up to and on the Cauchy horizon, this indicates that this radiation will also remain finite on the Cauchy horizon (with the exception of a possible divergence at the past endpoint of the naked singularity). 

One deficiency of the current work is the choice of initial data surface. The surface $z_{i}=0$ intersects the past end point of the naked singularity; it would be preferable to have a surface $t=t_{1}$ which intersects the regular centre of the spacetime prior to the formation of the naked singularity. Given the results already shown, the challenge here would be to show that regular initial data on such a surface evolves to regular data on the surface $z=z_{i}$. The results already proven then show that this data remains finite on the Cauchy horizon. We are currently considering this problem. 

The results in this paper do not support the hypothesis of cosmic censorship, which might have encouraged an expectation that such perturbations would diverge on the Cauchy horizon associated with the naked singularity. The finiteness of the perturbation also suggests that it may be possible to continue the spacetime evolution past the Cauchy horizon. However, we note that a full study of the behaviour of linear perturbations of this spacetime would include the even parity perturbations, which were set to zero here. The question of stability of the Cauchy horizon to linear perturbations therefore cannot be fully answered here. 

The even parity perturbations of this spacetime obey a much more complex system of differential equations, which can be studied using energy methods broadly similar to those employed here, as well as methods for systems of ODEs. The results of our study of the even parity perturbations of this spacetime will be presented in a future paper. 

As noted in the introduction, the dust spacetime is not an entirely accurate model of gravitational collapse, as it neglects pressure and pressure gradients. Another interesting extension would be to examine perturbations of the self-similar perfect fluid model, in which the fluid has a non-zero pressure. 

\ack
This research was funded by the Irish Research Council for Science, Engineering and Technology, grant number P02955.


\appendix
\section{}
\label{appendixcoeffs}
The coefficients of the master equation (\ref{master}) are given by
\begin{equation} \label{mastercoeff1}
\beta(z)=1-z^2 \rme^{-\nu}, 
\end{equation}
\begin{equation} \label{mastercoeff2}
\gamma(z)=-\rme^{-\nu}, 
\end{equation}
\begin{equation} \label{mastercoeff3}
\xi(z)=2z \rme^{-\nu}, 
\end{equation}
\begin{equation} \label{mastercoeff4}
a(z)=2z \rme^{-\nu}(2-\kappa) + \frac{\dot{\nu}}{2}(1+z^2 \rme^{-\nu})-\frac{2\dot{S}}{S}\beta(z), 
\end{equation}
\begin{equation} \label{mastercoeff5}
b(z)=\rme^{-\nu}(2 \kappa -5)-\rme^{-\nu}z\left( \frac{\dot{\nu}}{2} + \frac{2 \dot{S}}{S} \right), 
\end{equation}
\begin{equation} \label{mastercoeff6}
c(z)=-\rme^{-\nu}(\kappa^2-5\kappa+4)+z \rme^{-\nu}\left( \frac{\dot{\nu}}{2} + \frac{2 \dot{S}}{S} \right)(\kappa -4) + \mathcal{L}S^{-2}. 
\end{equation}

\section{}
\label{appendixa}

The matrices $X(z)$ and $W(z)$ appearing in (\ref{firstorderred}) are given by
\[X= \left( \begin{array}{ccc}
0&0&0 \\
0&0&-\gamma(z) \\
0&1&-\xi(z) \end{array} \right),\] 

\[W= \left( \begin{array}{ccc}
0&1&-\xi(z)\\
-c(z)&\left(-\frac{\beta,_{\bar{z}}}{\beta(z)}+a(z) \right)&-\xi(z)\left(-\frac{\beta,_{\bar{z}}}{\beta(z)}+a(z) \right)+(\xi,_{\bar{z}}(z)-b(z)) \\
0&0&0 \end{array} \right).\] 
The source vector $\vec{j}$ is given by $\vec{j} = (0, \rme^{\kappa p} \Sigma(z,p), 0)^T$. 

In order to use the standard theorem which proves existence and uniqueness of solutions to systems sych as (\ref{firstorderred}), we require $X$ and $W$ to be smooth, matrix-valued bounded functions of $\bar{z}$ on $[0, \infty)$, such that $X$ is symmetric with real, distinct eigenvalues. The matrix $X$ given above is not symmetric, however, it is easy to check that it is symmetrizable, and therefore a first order symmetric hyperbolic form of (\ref{master}) does exist. The matrix which symmetrizes $X$ is 

\[N= \left( \begin{array}{ccc}
1&0&0\\
0& \rme^{\nu}(z-\rme^{\nu/2}(1+z^2 \rme^{-\nu})^{1/2})& \rme^{-\nu}(z+\rme^{\nu/2}(1+z^2 \rme^{-\nu})^{1/2}) \\
0&1&1 \end{array} \right),\] 
so that $\tilde{X}=N^{-1}XN$ is a symmetric matrix. $\tilde{X}(z)$ and $\tilde{W}(z)=N^{-1}WN$ are given by

\[\tilde{X}= \left( \begin{array}{ccc}
0&0&0 \\
0&x_{+}(z)&0 \\
0&0&x_{-}(z) \end{array} \right),\] 
where 
\begin{equation*}
x_{\pm}(z)=\frac{\pm 4 \gamma(z) \mp \xi(z) (\xi(z) + \sqrt{-4 \gamma(z) + \xi^2(z)})}{2 \sqrt{-4 \gamma(z) + \xi^2(z)}},
\end{equation*}
and 
\begin{equation}
\label{Wmatrix}
\tilde{W}= \left( \begin{array}{ccc}
0&y_{1}^{+}(z)&y_{1}^{-}(z)\\
\frac{c(z)}{\sqrt{-4 \gamma(z) +\xi ^2(z)}} & y_{2}^{+}(z) & y_{2}^{-}(z)\\
-\frac{c(z)}{\sqrt{-4 \gamma(z) +\xi ^2(z)}}&y_{3}^{+}(z)&y_{3}^{-}(z)
\end{array} \right). 
\end{equation}
The components of the $\tilde{W}$ matrix are given by
\begin{equation*}
y_{1}^{\pm}(z)=\frac{1}{2} \left(-\xi \mp \sqrt{-4 \gamma(z) +\xi ^2(z)}\right),
\end{equation*}

\begin{equation*}
y_{2}^{\pm}(z)=\frac{\zeta(z)}{2 \beta(z)  \sqrt{-4 \gamma(z) +\xi ^2(z)}},
\end{equation*}
\begin{equation*}
y_{3}^{\pm}(z) = \frac{w(z)}{2 \sqrt{-4 \gamma(z) +\xi(z) ^2}}, 
\end{equation*}
where 
\begin{eqnarray*}
\zeta(z)=3 \xi(z) \pm \sqrt{-4 \gamma(z) +\xi ^2(z)} \dot{\beta}(z) \\ 
\qquad \qquad \qquad +\beta(z)  (2 b(z)-3 a(z) \xi(z) +a(z) \sqrt{-4 \gamma(z) +\xi ^2(z)}-2 \dot{\xi}(z)), 
\end{eqnarray*}
and
\begin{eqnarray*}
w(z)=(\xi(z) \mp \sqrt{-4 \gamma(z) +\xi ^2(z)}) \left(a-\dot{\beta}(z)/\beta(z) \right) \\
\qquad \qquad \qquad  -2 b(z)-2 \xi(z) \left(-a(z)+\dot{\beta}(z)/\beta(z)\right)+2 \dot{\xi}(z).
\end{eqnarray*}
The symmetric hyperbolic form of (\ref{master}) is given by
\begin{equation}
\label{redsymm}
\vec{\Psi},_{\bar{z}} = \tilde{X} \vec{\Psi},_{p} + (N^{-1}_{\bar{z}}N + \tilde{Y}) \vec{\Psi} + \vec{j}'
\end{equation}
where $\vec{\Psi}:=N^{-1} \vec{\Phi}$, and $\vec{j}'=N^{-1} \vec{j}$ is given by

\[\vec{j}'= \left( \begin{array}{c}
0 \\
-\frac{1}{2} \rme^{\nu/2} \rme^{\kappa p} \Sigma \\
\frac{1}{2} \rme^{\nu/2} \rme^{\kappa p} \Sigma \end{array} \right).\] 

\section{}
\label{appendixb}
In this appendix we provide the proofs of lemmas \ref{Lem3}, \ref{Lem4} and \ref{Lem5} which were omitted in the main text. 
\\ 
\\
\textbf{Proof of Lemma 5.1:}
That $A(u,v)$ is bounded follows immediately from the results in section \ref{sec:chbehaviour}. To bound $A,_{u}(u,v)$ and $A,_{v}(u,v)$, we write them in terms of $A,_{z}$ and $A,_{p}$. We find that 

\begin{equation*} 
A,_{u}(u,v)= \frac{f_{+}(z)}{u} \left( \frac{f_{-}(z)}{f_{-}(z)-f_{+}(z)} \right) \left[ \frac{\partial A}{\partial z} -\frac{1}{f_{-}(z)} \frac{\partial A}{\partial p} \right], 
\end{equation*}
\begin{equation}
\label{vderiv}
A,_{v}(u,v)= \frac{1}{v} \left( \frac{f_{-}(z)}{f_{+}(z)-f_{-}(z)} \right) \left[ f_{+}(z)\frac{\partial A}{\partial z} -\frac{\partial A}{\partial p} \right]. 
\end{equation}
We note that by using (\ref{uvcoords}), one can show that $\frac{f_{+}(z)}{u}$ tends to a finite value as $z \rightarrow z_{c}$. Then since $A,_{z}$ and $A,_{p}$ can be bounded by \textit{a priori} initial data (see (\ref{azbound}) and (\ref{apbound})), it follows that $A,_{u}(u,v)$ can be bounded by similar \textit{a priori} terms. By an exactly similar argument, we can show that $A,_{v}$ is bounded. 
\hfill$\square$
\\
\\
\textbf{Proof of Lemma 5.2:}
We can write (\ref{masteruv}) as $A,_{uv}=F(u,v)$ and by noting the form of the coefficients (\ref{uvcoefficients}) and that $A,_{u}$ and $A,_{v}$ are bounded, it follows that $A,_{uv}$ is bounded in the approach to the Cauchy horizon. 

To deal with $A,_{vv}$, we first write (\ref{vderiv}) as $A,_{v}=\frac{H(z,p)}{v}$, where 
\begin{equation*}
H(z,p)=\frac{f_{-}}{f_{+}-f_{-}} (f_{+} A,_{z} - A,_{p}). 
\end{equation*}
Taking a derivative with respect to $v$ and converting to $(z,p)$ coordinates produces a set of terms which depend on $H$, $H,_{p}$ and $H,_{z}$,
\begin{equation}
\label{vvderiv}
A,_{vv}=-\frac{1}{v^2} H(z,p) + \frac{1}{v^2} \left( f_{-} \frac{\partial H}{\partial z} + \frac{\partial H}{\partial p} \right). 
\end{equation}
The partial derivatives of $H(z,p)$ with respect to $z$ and $p$ are given by
\begin{equation*}
\frac{\partial H}{\partial p} = \frac{f_{-}}{f_{+}-f_{-}} (f_{+} A,_{zp} - A,_{pp}),
\end{equation*}
and 
\begin{equation}
\label{zderivH}
\frac{\partial H}{\partial z} = \left( \frac{f_{-}}{f_{+}-f_{-}} \right),_{z} (f_{+} A,_{z} - A,_{p}) + \frac{f_{-}}{f_{+}-f_{-}} (f_{+},_{z} A,_{z} + f_{+} A,_{zz} - A,_{zp}). 
\end{equation}
Now in (\ref{vvderiv}), the terms involving $H$ and $H,_{p}$ remain finite in the approach to the Cauchy horizon. This follows since these terms involve $A,_{z}$, $A,_{p}$, $A,_{pp}$ and $A,_{zp}$, which Theorem \ref{Thm3} and Lemma \ref{Lem2} show to be bounded by \textit{a priori} terms (recall that we can take $p$-derivatives in these results to show that $A,_{p}$, $A,_{pp}$ and $A,_{zp}$ are bounded). 

It remains to show that the $H,_{z}$ term remains finite as the Cauchy horizon is approached. If we examine (\ref{zderivH}), we see that we have terms involving $A,_{z}$, $A,_{p}$, $A,_{zp}$ and $f_{+}A,_{zz}$. The first three of these remain finite in the approach to the Cauchy horizon as discussed above. Now by solving (\ref{master}) for $A,_{zz}$, we find that we can write $f_{+} A,_{zz}$ as 
\begin{equation}
\label{zzderivA}
\fl f_{+}A,_{zz}= - \frac{f_{+}}{\beta(z)}(\gamma(z) A,_{pp} + \xi(z)A,_{zp} + a(z)A,_{z} + b(z)A,_{p} + c(z) A - \rme^{\kappa p} \Sigma(z,p)). 
\end{equation}
The term $(\gamma(z) A,_{pp} + \xi(z)A,_{zp} + a(z)A,_{z} + b(z)A,_{p} + c(z) A - \rme^{\kappa p} \Sigma(z,p))$ remains finite in the approach to the Cauchy horizon. This follows immediately from the boundedness of $A,_{pp}$, $A,_{zp}$, $A,_{z}$, $A,_{p}$ and $A$, and from the boundedness of the coefficients $\gamma(z)$, $\xi(z)$, $a(z)$, $b(z)$ and $c(z)$. The source term $\rme^{\kappa p} \Sigma(z,p)$ is bounded everywhere, assuming a finite perturbation of the dust velocity, see (\ref{Sigmadef}). 

To deal with the factor of $f_{+} \beta^{-1}$ in (\ref{zzderivA}), we note that $\beta(z)$ can be written as $\beta(z)=\rme^{-\nu}(\rme^{\nu/2} - z)(\rme^{\nu/2} +z)$ (see \ref{mastercoeff1}), so that $f_{+} \beta^{-1} = \rme^{\nu}(\rme^{\nu/2} - z)^{-1}$. This remains finite as the Cauchy horizon is approached. We can therefore conclude that the term $f_{+} A,_{zz}$ remains finite, and indeed, is bounded by \textit{a priori} terms. It follows that (\ref{zderivH}) also remains finite, and finally $A,_{vv}$, given by (\ref{vvderiv}) also remains finite as the Cauchy horizon is approached, and is bounded by \textit{a priori} terms inherited from the bounds on $A$, $A,_{z}$, $A,_{p}$, $A,_{pp}$, $A,_{zp}$ and $\Sigma$. 
%

Finally, we must show that $A,_{uu}$ is bounded in the approach to the Cauchy horizon. This proof involves two steps. Firstly, by following a procedure similar to that given above, we can establish a bound on $u A,_{uu}$. We write $A,_{u}$ as $A,_{u}=\frac{G(z,p)}{u}$ where 
\begin{equation*}
G(z,p)=\frac{f_{+}f_{-}}{f_{-}-f_{+}} \left( A,_{z} - \frac{A,_{p}}{f_{-}} \right). 
\end{equation*}
Taking a $u$ derivative of $A,_{u}$ results in 
\begin{equation}
\label{uuderivA}
A,_{uu}= -\frac{1}{u^2} G(z,p) + \frac{1}{u^2} \left( f_{+} \frac{\partial G}{\partial z} + \frac{\partial G}{\partial p} \right). 
\end{equation}
$G,_{z}$ and $G,_{p}$ are given by
\begin{equation}
\label{pderivG}
\frac{\partial G}{\partial p}=\frac{f_{+}f_{-}}{f_{-}-f_{+}} \left( A,_{zp} - \frac{A,_{pp}}{f_{-}} \right),
\end{equation}
and 
\begin{equation}
\label{zderivG}
\frac{\partial G}{\partial z} = \left( \frac{f_{+}f_{-}}{f_{-}-f_{+}} \right),_{u} \left( A,_{z} - \frac{A,_{p}}{f_{-}} \right) + \frac{f_{+}f_{-}}{f_{-}-f_{+}} \left( A,_{zz} - \frac{A,_{zp}}{f_{-}} + \frac{f_{-},_{z}}{f_{-}^2} A,_{p} \right). 
\end{equation}
By the same argument as that given above, we can conclude that $f_{+}A,_{zz}$ remains finite as the Cauchy horizon is approached, so that $G,_{z}$ and $G,_{p}$ both remain finite in this limit. If we combine (\ref{uuderivA}), (\ref{pderivG}) and (\ref{zderivG}), we find that we can write $A,_{uu}$ as 
\begin{equation*}
A,_{uu}=\frac{f_{+}}{u^2} \tilde{G}(z,p), 
\end{equation*}
where 
\begin{equation*}
\tilde{G}(z,p)=\frac{\partial G}{\partial z} + \frac{f_{-}}{f_{-}-f_{+}} \left(A,_{zp} - A,_{z} + \frac{A,_{p} - A,_{pp}}{f_{-}} \right), 
\end{equation*}
and this term is bounded by \textit{a priori} terms as the Cauchy horizon is approached. Now since $\frac{f_{+}}{u}$ tends to a finite constant as $z \rightarrow z_{c}$, it follows that $u A,_{uu}$ is bounded by \textit{a priori} terms inherited from the bounds on $A$, $A,_{z}$, $A,_{p}$, $A,_{pp}$, $A,_{zp}$ and $\Sigma$ . This result is not sufficient for our requirements in section \ref{sec:giinterp}, but we can use it to establish a stronger bound on $A,_{uu}$. 

We return to the wave equation, in the form $A,_{uv}=F(u,v)$. By integrating with respect to $v$, we find
\begin{equation}
\label{uderiv}
A,_{u}(u,v) - A,_{u}(u, v_{0}) = \int_{v_{0}}^{v} F(u, \bar{v}) d \bar{v}. 
\end{equation}
where $F(u,v)=(\alpha_{1})^{-1} \left(-\alpha_{2}(u,v) u \, A,_{u} - \alpha_{3}(u,v) v \, A,_{v} - \alpha_{4}(u,v) A + \rme^{\kappa p} \Sigma(u,v) \right)$. Now we can make a choice of $v_{0}$ such that $v_{0}$ does not intersect the compact support of the perturbation (see figure \ref{Fig3}). With this choice, $A,_{u}(u, v_{0}) =0$. Then taking a $u$ derivative of (\ref{uderiv}) results in 
\begin{equation}
\label{uuderivAint}
A,_{uu} = \int_{v_{0}}^{v} \frac{\partial F}{\partial u}(u, \bar{v}) d \bar{v}. 
\end{equation}
We can write $\frac{\partial F}{\partial u}$ as 
\begin{eqnarray}
\label{uderivF}
\fl \frac{\partial F}{\partial u} = \frac{\alpha_{1},_{u}}{\alpha_{1}^2} \left(-\alpha_{2}(u,v) u \, A,_{u} - \alpha_{3}(u,v) v \, A,_{v} - \alpha_{4}(u,v) A + \rme^{\kappa p} \Sigma(u,v) \right)  \\ \nonumber
\fl \qquad \qquad - \frac{1}{\alpha_{1}} \left( \alpha_{2},_{u} u A,_{u} + \alpha_{2} A,_{u} + \alpha_{2} u A,_{uu} + \alpha_{3},_{u}vA,_{v} + \alpha_{3} v A,_{uv} + \alpha_{4},_{u} A + \alpha_{4} A,_{u} - (\rme^{\kappa p} \Sigma),_{u}  \right). 
\end{eqnarray}
Now, if we consider (\ref{uvcoefficients}), we see that we can write $\alpha_{1}=\frac{\tilde{\alpha}_{1}}{f_{+}}$, where $\tilde{\alpha}_{1}= 2z(\beta + \xi)f_{-}^{-1} + 2 \gamma f_{+}$. Then we can show that 
\begin{equation}
\label{coeff1}
\frac{\alpha_{1},_{u}}{\alpha_{1}^2}= \frac{f_{+}}{u} \frac{1}{\tilde{\alpha}_{1}^2} (-\tilde{\alpha}_{1} f_{+},_{z} + \tilde{\alpha}_{1},_{z} f_{+}),
\end{equation}
which remains finite as the Cauchy horizon is approached (since $\frac{f_{+}}{u}$ tends to a finite constant as $z \rightarrow z_{c}$, and by inspection, so do $\tilde{\alpha}_{1}$, $\tilde{\alpha}_{1},_{z}$ and $f_{+},_{z}$). From (\ref{coeff1}), and the results of Lemma \ref{Lem3}, we can  conclude that the first term in (\ref{uderivF}) remains finite as the Cauchy horizon is approached. 

Similarly, by writing $\alpha_{2}=\frac{\tilde{\alpha}_{2}}{f_{+}}$, where $\tilde{\alpha}_{2}= a(z)+b(z)f_{+}$, we can show that 
\begin{equation}
\label{coeff2}
\frac{\alpha_{2},_{u}}{\alpha_{1}} u = \frac{1}{\tilde{\alpha}_{1}} (-\tilde{\alpha}_{2} f_{+},_{z} + f_{+} \tilde{\alpha}_{2},_{z}),
\end{equation}
and 
\begin{equation}
\label{coeff3}
\frac{\alpha_{2}}{\alpha_{1}} = \frac{\tilde{\alpha}_{2}}{\tilde{\alpha}_{1}},
\end{equation}
which remain bounded as the Cauchy horizon is approached, since by inspection, $\tilde{\alpha}_{1}$, $\tilde{\alpha}_{2}$, $f_{+},_{z}$ and $\tilde{\alpha}_{2},_{z}$ tend to finite constants as the Cauchy horizon is approached. Now the terms $\alpha_{3},_{u}vA,_{v} + \alpha_{4},_{u} A + \alpha_{4} A,_{u} - (\rme^{\kappa p} \Sigma),_{u}$ which appear in the second term in (\ref{uderivF}) are bounded as the Cauchy horizon is approached. This follows from Lemma \ref{Lem3} and from the boundedness of the coefficients $\alpha_{3},_{u}$, $\alpha_{4},_{u}$, $\alpha_{4}$ (which can be easily seen by converting to $(z,p)$ coordinates and using \ref{uvcoefficients}).  $(\rme^{\kappa p} \Sigma),_{u}$ reduces to $z$ and $p$ derivatives of $\Sigma$ when we switch to $(z,p)$ coordinates, and these are bounded in the approach to the Cauchy horizon. 

The remaining terms from (\ref{uderivF}) which must be dealt with are $\alpha_{2},_{u} u A,_{u}\alpha_{1}^{-1}$, $\alpha_{2} A,_{u}\alpha_{1}^{-1}$, $ \alpha_{2} u A,_{uu}\alpha_{1}^{-1}$, $ \alpha_{3} v A,_{uv}\alpha_{1}^{-1}$. Now combining (\ref{coeff2}), (\ref{coeff3}) and the results of Lemma \ref{Lem3}, as well as our result for the boundedness of $u A,_{uu}$, we see that these terms remain bounded as the Cauchy horizon is approached. We can therefore conclude that $\frac{\partial F}{\partial u}$ remains finite, and therefore, from (\ref{uuderivAint}), $A,_{uu}$ also remains bounded as the Cauchy horizon is approached. In particular, $A,_{uu}$ is bounded by \textit{a priori} terms arising from the bounds on $A$, $A,_{z}$, $A,_{p}$, $A,_{pp}$, $A,_{zp}$ and $\Sigma$. 
\hfill$\square$
\\
\\
\textbf{Proof of Lemma 5.3:}
That $A$ itself is bounded with  this choice of initial data follows immediately from the second part of Theorem \ref{Thm5}. To show the boundedness of the derivatives, we follow a procedure similar to that used to prove Theorem \ref{Thm5}. 

The space $C_{0}^{\infty}(\mathbb{R}, \mathbb{R})$ is dense in each of $\mathbb{H}^{3,2}(\mathbb{R}, \mathbb{R})$ and $\mathbb{H}^{1,2}(\mathbb{R}, \mathbb{R})$. It follows that there exist sequences $\{ f_{(m)} \}_{m=0}^{\infty}$, $\{ j_{(m)} \}_{m=0}^{\infty}$ and $\{ h_{(m)} \}_{m=0}^{\infty}$, with $f_{(m)}$, $j_{(m)}$ and $h_{(m)} \in C_{0}^{\infty}(\mathbb{R}, \mathbb{R})$, such that $f_{(m)} \rightarrow f$, $j_{(m)} \rightarrow j$ and $h_{(m)} \rightarrow h$ as $m \rightarrow \infty$, with convergence in the $\mathbb{H}^{3,2}(\mathbb{R}, \mathbb{R})$ and $\mathbb{H}^{1,2}(\mathbb{R}, \mathbb{R})$ norms respectively. 

Then for all $m \geq 0$, we take $f_{(m)}$, $j_{(m)}$ and $h_{(m)}$ as initial data for $A$ and $A,_{z}$, and apply Theorems \ref{Thm1}, \ref{Thm2} and \ref{Thm3} to find a sequence of solutions $A_{(m)}$ which obeys at each $m$ the \textit{a priori} bounds from Theorem \ref{Thm3} and Lemma \ref{Lem2}. By taking $p$-derivatives in these results as required, we can establish similar bounds on $A_{(m)},_{p}$, $A_{(m)},_{pp}$ and $A_{(m)},_{zp}$. We can then take the $m \rightarrow \infty$ limit and will find a solution $A \in C([z_{c}, z], \mathbb{H}^{3,2}(\mathbb{R}, \mathbb{R}))$. 

Finally, we apply lemmas \ref{Lem3} and \ref{Lem4} to establish bounds on the required $u$ and $v$ derivatives of $A$. These terms will be bounded by \textit{a priori} terms inherited from the bounds on $A$, $A,_{z}$, $A,_{p}$, $A,_{pp}$, $A,_{zp}$ and $\Sigma$. The order of derivatives of $A$ and $\Sigma$ required for these bounds dictates the choice of Sobolev spaces for the initial data specified in this lemma. 
\hfill$\square$


\end{document}